\documentclass[10pt,letterpaper,english,10pt,final,journal,letterpaper,onecolumn]{IEEEtran}
\usepackage[T1]{fontenc}
\usepackage[latin9]{inputenc}
\usepackage{amsthm}
\usepackage{amsmath}
\usepackage{amssymb}
\usepackage{graphicx}
\usepackage{esint}
\usepackage{epstopdf}

\makeatletter

  \theoremstyle{definition}
  \newtheorem{defn}{\protect\definitionname}
  \theoremstyle{plain}
  \newtheorem{thm}{\protect\theoremname}
  \theoremstyle{plain}
  \newtheorem{lem}{\protect\lemmaname}
  \theoremstyle{plain}
  \newtheorem{fact}{\protect\factname}
  \theoremstyle{remark}
  \newtheorem{rem}{\protect\remarkname}
  \theoremstyle{definition}
  \newtheorem{example}{\protect\examplename}

\IEEEoverridecommandlockouts

\usepackage{dsfont}\usepackage{babel}
\usepackage{relsize}\usepackage{subdepth}\allowdisplaybreaks

\author{
    \IEEEauthorblockN{Dionysios S. Kalogerias, \IEEEmembership{Student Member, IEEE}
             and Athina P. Petropulu$^{\ast}$, \IEEEmembership{Fellow, IEEE}}
   \thanks{D. S. Kalogerias and A. P. Petropulu are with the Department of Electrical \& Computer Engineering, Rutgers, The State University of New Jersey, 94 Brett Rd, Piscataway, 08854, NJ, USA. e-mail: \{d.kalogerias, athinap\}@rutgers.edu}
    \thanks{This work is supported by the NSF under Grant CNS-1239188, and the ONR under grant N00014-12-1-0036.}
}

\usepackage{dsfont}

\usepackage{babel}

\usepackage{relsize}\usepackage{subdepth}

\usepackage{babel}
\providecommand{\definitionname}{Definition}
\providecommand{\lemmaname}{Lemma}

\providecommand{\theoremname}{Theorem}

\usepackage{babel}
\providecommand{\definitionname}{Definition}
\providecommand{\examplename}{Example}
\providecommand{\factname}{Fact}
\providecommand{\lemmaname}{Lemma}
\providecommand{\remarkname}{Remark}
\providecommand{\theoremname}{Theorem}

\@ifundefined{showcaptionsetup}{}{%
 \PassOptionsToPackage{caption=false}{subfig}}
\usepackage{subfig}
\makeatother

\usepackage{babel}
\providecommand{\definitionname}{Definition}
\providecommand{\examplename}{Example}
\providecommand{\factname}{Fact}
\providecommand{\lemmaname}{Lemma}
\providecommand{\remarkname}{Remark}
\providecommand{\theoremname}{Theorem}

\begin{document}

\title{Matrix Completion in Colocated MIMO Radar:\\
 Recoverability, Bounds \& Theoretical Guarantees}
\maketitle
\begin{abstract}
It was recently shown that low rank matrix completion theory can be
employed for designing new sampling schemes in the context of MIMO
radars, which can lead to the reduction of the high volume of data
typically required for accurate target detection and estimation. Employing
random samplers at each reception antenna, a partially observed version
of the received data matrix is formulated at the fusion center, which,
under certain conditions, can be recovered using convex optimization.
This paper presents the theoretical analysis regarding the performance
of matrix completion in colocated MIMO radar systems, exploiting the
particular structure of the data matrix. Both Uniform Linear Arrays
(ULAs) and arbitrary $2$-dimensional arrays are considered for transmission
and reception. Especially for the ULA case, under some mild assumptions
on the directions of arrival of the targets, it is explicitly shown
that the coherence of the data matrix is both asymptotically and approximately
optimal with respect to the number of antennas of the arrays involved
and further, the data matrix is recoverable using a subset of its
entries with minimal cardinality. Sufficient conditions guaranteeing
low matrix coherence and consequently satisfactory matrix completion
performance are also presented, including the arbitrary $2$-dimensional
array case. \end{abstract}

\section{Introduction}

\IEEEPARstart{L}{ow} rank matrix completion aims at the reconstruction
of a low rank data matrix using a set of limited observations of its
entries. Recently, this type of estimation problems have received
a lot of attention, because they arise in a variety of real world
applications, such as collaborative filtering \cite{NetFlixPrize},
sensor network positioning/localization \cite{Drineas2006} and remote
sensing \cite{SunMIMORadarMC}, just to name a few. Several theoretical
results have been reported \cite{Candes&Recht2009,Candes&Tao2010,KeshavanMontanariMC},
and algorithms have been proposed (see, for example, \cite{CaiCandesSVT}),
providing the essential tools for solving the matrix completion problem
reliably and robustly.

Most recently, matrix completion has been proposed as means for effectively
reducing the volume of data required for target detection and estimation
in MIMO radars \cite{HaimovichRadarIDEA2004} and more generally in
array processing systems \cite{WengMC_ARRAYS,SunMIMORadarMC,SunPetropuluAEROSPACE}.
In \cite{SunPetropuluAEROSPACE}, a collocated MIMO radar \cite{LiStoicaColocated2007}
scenario is considered, in which the transmission and reception antennas
are organized in Uniform Linear Arrays (ULAs). Each transmission antenna
transmits a narrowband waveform over a predefined carrier frequency,
with the waveforms between different transmission antennas being orthogonal.
At each reception antenna, after demodulation, matched filtering is
performed with the transmit waveforms \cite{SunPetropuluAEROSPACE},
extracting statistics and formulating a matrix (referred to here as
the data matrix), which can be used by standard array processing methods
for target detection and parameter estimation. For a sufficiently
large number of transmission and reception antennas and a small number
of targets, the data matrix is low-rank. Therefore, it can be recovered
from a small number of its entries via matrix completion. This implies
that, at each reception antenna, matched filtering does not need to
be performed with all transmit waveforms, but rather with a small
number of randomly selected ones from the waveform dictionary. The
general conditions for the applicability of matrix completion were
stated in \cite{SunMIMORadarMC,SunPetropuluAEROSPACE} and their validity
was confirmed via simulations.

In this paper, we consider the matrix completion enabled MIMO radar
system proposed in \cite{SunPetropuluAEROSPACE}, and exploiting the
special structure of the data matrix, we derive insightful theoretical
results regarding its \textit{coherence}, an important quantity closely
related to its recoverability via matrix completion. Our contribution
is summarized as follows: 
\begin{enumerate}
\item We show that, for ULA configurations, and under mild assumptions on
the Directions-Of-Arrival (DOAs) of the targets, the coherence of
the data matrix is both asymptotically optimal with respect to the
number of transmission and reception antennas, and nearly optimal
for a sufficiently large but finite number of transmission and reception
antennas. 
\item Under common assumptions regarding the range of the pairwise differences
of the target angles and the spacing of the antennas of the ULAs involved,
we derive a simple sufficient condition, which essentially controls
the coherence of the data matrix, as well as the rate of convergence
to its optimum value. In all cases, we provide explicit and computationally
tractable coherence bounds, with all results holding \textit{almost
surely}. 
\item Invoking the recent theoretical results in low rank matrix completion
presented in \cite{Candes&Tao2010}, we derive asymptotic bounds on
the number of observations required for exact matrix completion, showing
that, in fact, the matrix under consideration can be reconstructed
using a subset of its entries with minimal cardinality. 
\item We then generalize our coherence results for the more general case
of arbitrary $2$-dimensional transmission and reception arrays, showing
that the coherence of the corresponding data can be easily and tightly
bounded by simply looking at the values of the squared magnitude of
a $2$-dimensional complex function over a properly chosen subset
of the $x/y$ plane, directly related to the target angles. Essentially,
by choosing a candidate pair of transmission/reception array topologies,
this general result produces coherence bounds in a straightforward
way, also holding true almost surely. 
\end{enumerate}
An important implication of the results presented in this paper is
that the performance of matrix completion in colocated MIMO radar
systems is completely independent of any other system design specification
or target characteristic, except for the transmission/reception array
topology, the target angles and possibly the wavelength of the carrier.
That is, as far as matrix completion is concerned, the values of various
important quantities, such as target reflection coefficients, Doppler
shifts, pulse repetition intervals and the choice of the waveform
dictionary are \textit{completely unconstrained}. This fact makes
matrix completion very appealing for reducing the number of samples
needed for accurate detection and estimation in real world colocated
MIMO radar systems. Also, compared to Compressive Sensing (CS) based
MIMO radar, matrix completion has at least the same advantage in terms
of sample reduction, justified by the strong theoretical guarantees
presented here, while at the same time it avoids the angle discretization
and basis mismatch issues inherent in CS-based approaches \cite{YuPetropuluPoor2010,NehoraiMIMORadar2011,Rossi2013}.

The paper is organized as follows. In Section II, we briefly introduce
the required background in both noiseless and noisy matrix completion.
We also restate the problem formulation for the MIMO radar case in
a slightly more general form compared to \cite{SunPetropuluAEROSPACE}.
In Section III, we present our coherence and recoverability results
for ULA transmitter/receiver pairs, complete with detailed proofs
and explanations. In Section IV, we present the generalization to
the aforementioned coherence results for arbitrary $2$-dimensional
transmission/reception arrays, along with an instructive example.
In Section IV, we state and prove a result for some special cases
of spatial target configurations, completing our theoretical analysis.
Finally, in Section V, we discuss further implications of our previously
stated results and we present some simulations, validating their correctness.

\section{Background}

\subsection{Noiseless Matrix Completion:\protect \protect \\
 Problem Statement and Recoverability Conditions}

Consider a generic complex matrix $\mathbf{M}\in\mathbb{C}^{N_{1}\times N_{2}}$
of rank $r$, whose compact Singular Value Decomposition (SVD) is
given by $\mathbf{M}=\mathbf{U}\boldsymbol{\Sigma}\mathbf{V}^{\boldsymbol{\mathit{H}}}\equiv\sum_{i\in\mathbb{N}_{r}^{+}}\sigma_{i}\left(\mathbf{M}\right)\mathbf{u}_{i}\mathbf{v}_{i}^{\boldsymbol{\mathit{H}}}$
and with column and row subspaces denoted as $U$ and $V$ respectively,
spanned by the sets $\left\{ \mathbf{u}_{i}\in\mathbb{C}^{N_{1}\times1}\right\} _{i\in\mathbb{N}_{r}^{+}}$
and $\left\{ \mathbf{v}_{i}\in\mathbb{C}^{N_{2}\times1}\right\} _{i\in\mathbb{N}_{r}^{+}}$,
respectively.

Let $\mathcal{P}\left(\mathbf{M}\right)\in\mathbb{C}^{N_{1}\times N_{2}}$
denote an entrywise sampling of $\mathbf{M}$. In all the analysis
that follows, we will adopt the theoretical frameworks presented in
\cite{Candes&Recht2009} and \cite{Candes&Tao2010}, according to
which one hopes to reconstruct $\mathbf{M}$ from $\mathcal{P}\left(\mathbf{M}\right)$
by solving the convex program 
\begin{equation}
\begin{array}{ll}
\mathrm{minimize} & \left\Vert \mathbf{X}\right\Vert _{*}\\
\mathrm{subject\, to} & \mathbf{X}\left(i,j\right)=\mathbf{M}\left(i,j\right),\quad\forall\,\left(i,j\right)\in\boldsymbol{\Omega}
\end{array},\label{eq:convex1}
\end{equation}
where the set $\boldsymbol{\Omega}$ contains all matrix coordinates
corresponding to the observed entries of $\mathbf{M}$ (contained
in $\mathcal{P}\left(\mathbf{M}\right)$) and where $\left\Vert \mathbf{X}\right\Vert _{*}$
represents the nuclear norm of $\mathbf{X}$. In the following, we
will refer to \eqref{eq:convex1} as the Matrix Completion (MC) problem.

Also in \cite{Candes&Recht2009}, the authors introduce the notion
of \textit{subspace coherence}, in order to derive specific conditions
under which the solution of \eqref{eq:convex1} coincides with $\mathbf{M}$.
The formal definition of subspace coherence follows, in a slightly
more expanded form compared to the original definition stated in \cite{Candes&Recht2009}. 
\begin{defn}
\textbf{(Subspace Coherence \cite{Candes&Recht2009})} Let $U\equiv\mathbb{C}^{r}\subseteq\mathbb{C}^{N}$
be a subspace spanned by the set of orthonormal vectors $\left\{ \mathbf{u}_{i}\in\mathbb{C}^{N\times1}\right\} _{i\in\mathbb{N}_{r}^{+}}$.
Also, define the matrix $\mathbf{U}\triangleq\left[\mathbf{u}_{1}\,\mathbf{u}_{2}\,\ldots\,\mathbf{u}_{r}\right]\in\mathbb{C}^{N\times r}$
and let $\mathbf{P}_{U}\triangleq\mathbf{U}\mathbf{U}^{\boldsymbol{\mathit{H}}}\in\mathbb{R}^{N\times N}$
be the orthogonal projection onto $U$. Then, the coherence of $U$
with respect to the standard basis $\left\{ \mathbf{e}_{i}\right\} _{i\in\mathbb{N}_{N}^{+}}$
is defined as 
\begin{flalign}
\mu\left(U\right) & \triangleq\frac{N}{r}\sup_{i\in\mathbb{N}_{N}^{+}}\left\Vert \mathbf{P}_{U}\mathbf{e}_{i}\right\Vert _{2}^{2}\nonumber \\
 & \equiv\frac{N}{r}\sup_{i\in\mathbb{N}_{N}^{+}}\sum_{k\in\mathbb{N}_{r}^{+}}\left|\mathbf{U}\left(i,k\right)\right|^{2}\in\left[1,\dfrac{N}{r}\right].
\end{flalign}

\end{defn}
As explained in \cite{Candes&Recht2009}, an intuitive motivation
for defining subspace coherence stems from the fact that, in some
sense, the singular vectors of the matrix under consideration need
to be ``sufficiently spread'', or more precisely, need to be uncorrelated
to the standard basis, in order for the matrix to be fully recoverable
using only a small number of its entries.

The following assumptions regarding the subspaces $U$ and $V$ are
of particular importance \cite{Candes&Recht2009}.\\
 \ \\
 $\mathbf{A0}\qquad\max\left\{ \mu\left(U\right),\mu\left(V\right)\right\} \leq\mu_{0}\in\mathbb{R}_{++}$.\\
 \ \\
 $\mathbf{A1}\qquad\left\Vert \sum_{i\in\mathbb{N}_{r}^{+}}\mathbf{u}_{i}\mathbf{v}_{i}^{\boldsymbol{\mathit{H}}}\right\Vert _{\infty}\leq\mu_{1}\sqrt{\dfrac{r}{N_{1}N_{2}}},\quad\mu_{1}\in\mathbb{R}_{++}$.\\
 \ \\
 Indeed, if the constants $\mu_{0}$ and $\mu_{1}$ associated with
the singular vectors of a matrix $\mathbf{M}$ are known, the following
theorem holds. 
\begin{thm}
\textbf{\textup{(Exact MC \cite{Candes&Recht2009})}} Let $\mathbf{M}\in\mathbb{C}^{N_{1}\times N_{2}}$
be a matrix of rank $r$ obeying $\mathbf{A0}$ and $\mathbf{A1}$
and set $N\triangleq\max\left\{ N_{1},N_{2}\right\} $. Suppose we
observe $m$ entries of $\mathbf{M}$ with locations sampled uniformly
at random. Then there exist constants $C$, $c$ such that if 
\begin{equation}
m\geq C\max\left\{ \mu_{1}^{2},\mu_{0}^{1/2}\mu_{1},\mu_{0}N^{1/4}\right\} Nr\beta\log N
\end{equation}
for some $\beta>2$, the minimizer to the program (\ref{eq:convex1})
is unique and equal to $\mathbf{M}$ with probability at least $1-cN^{-\beta}$.
For $r\leq\mu_{0}^{-1}N^{1/5}$ this estimate can be improved to 
\begin{equation}
m\geq C\mu_{0}N^{6/5}r\beta\log N,
\end{equation}
with the same probability of success. 
\end{thm}
If a matrix $\mathbf{M}$ obeys the assumptions $\mathbf{A0}$ and
$\mathbf{A1}$ with parameters $\mu_{0}$ and $\mu_{1}$, we will
say that it obeys the \textit{incoherence property} or, equivalently,
that it is \textit{incoherent} with parameters $\mu_{0}$ and $\mu_{1}$.

In \cite{Candes&Tao2010}, the respective authors introduce a stricter
assumption on the singular vectors of a generic matrix $\mathbf{M}$,
as a replacement for $\mathbf{A0}$, enabling them to prove tighter
-in fact, almost optimal- bounds on the number of observations required
in order to achieve exact reconstruction of $\mathbf{M}$ by solving
\eqref{eq:convex1}. Implicitly, they replace subspace coherence with
a closely related quantity, which we will refer to as \textit{strong
subspace coherence}, whose definition is presented below. 
\begin{defn}
\textbf{(Strong Subspace Coherence)} Consider the hypotheses of Definition
1. Then, the strong coherence of $U$ with respect to the standard
basis $\left\{ \mathbf{e}_{i}\right\} _{i\in\mathbb{N}_{N}^{+}}$
is defined as 
\begin{equation}
\mu_{s}\left(U\right)\triangleq\sup_{\left(i,j\right)\in\mathbb{N}_{N}^{+}\times\mathbb{N}_{N}^{+}}\left|\frac{N}{r}\left\langle \mathbf{e}_{i},\mathbf{P}_{U}\mathbf{e}_{j}\right\rangle -\mathds{1}_{i=j}\right|.
\end{equation}

\end{defn}
Using the definition stated above, the authors in \cite{Candes&Tao2010}
essentially replace $\mathbf{A0}$ by the following assumption, concerning
the singular vectors of $\mathbf{M}$.\\
 \ \\
 $\mathbf{A2}\qquad\max\left\{ \mu_{s}\left(U\right),\mu_{s}\left(V\right)\right\} \leq\dfrac{\mu_{0}^{s}}{\sqrt{r}},\quad\mu_{0}^{s}\in\mathbb{R}_{++}$.\\
 \ \\
 If $\mathbf{M}$ obeys the assumptions $\mathbf{A1}$ and $\mathbf{A2}$
with a parameter $\mu\ge\max\left\{ \mu_{0}^{s},\mu_{1}\right\} $,
we will say that it obeys the \textit{strong incoherence property}
or, equivalently, that it is \textit{strongly incoherent} with parameter
$\mu$.

In addition, the following $2$-in-$1$ theorem holds. 
\begin{thm}
\textbf{\textup{(Exact MC I \& II \cite{Candes&Tao2010})}} Let $\mathbf{M}\in\mathbb{C}^{N_{1}\times N_{2}}$
be a matrix of rank $r$ obeying the strong incoherence property with
parameter $\mu$ and set $N\triangleq\max\left\{ N_{1},N_{2}\right\} $.
Suppose we observe $m$ entries of $\mathbf{M}$ with locations sampled
uniformly at random. Then, there exist positive numerical constants
$C_{1}$ and $C_{2}$ such that if 
\begin{flalign}
m & \ge C_{1}\mu^{4}Nr^{2}\log^{2}N\quad\text{or}\\
m & \ge C_{2}\mu^{2}Nr\log^{6}N,
\end{flalign}
the minimizer to the program \eqref{eq:convex1} is unique and equal
to $\mathbf{M}$ with probability at least $1-N^{-3}$. 
\end{thm}
Further, it is easy to show that, in fact, incoherence implies strong
incoherence (stated relatively informaly in \cite{Candes&Tao2010}),
as the following lemma asserts. 
\begin{lem}
If a matrix $\mathbf{M}\in\mathbb{C}^{N_{1}\times N_{2}}$ of rank
$r$ is incoherent with parameters $\mu_{0}$ and $\mu_{1}$, it is
also strongly incoherent with parameter $\mu\le\mu_{0}\sqrt{r}$. 
\end{lem}
\begin{proof}[Proof of Lemma 1] See Appendix A.\end{proof}

\subsection{Noisy Matrix Completion: Stability}

In the previous subsection, we have focused on cases in which the
matrix observations are perfectly noiseless. Of course, in any realistic
setting, these observations will be corrupted by noise, yielding the
observation model (for the same generic matrix $\mathbf{M}$) 
\begin{equation}
\mathcal{P}\left(\mathbf{Y}\right)\triangleq\mathcal{P}\left(\mathbf{M}\right)+\mathcal{P}\left(\mathbf{Z}\right),
\end{equation}
where $\mathbf{Z}\in\mathbb{C}^{N_{1}\times N_{2}}$ constitutes a
noise matrix with $\left\Vert \mathcal{P}\left(\mathbf{Z}\right)\right\Vert _{F}\le\delta$
for some known constant $\delta>0$. If $\mathbf{Z}$ is random, then
$\delta$ is such that $\left\Vert \mathcal{P}\left(\mathbf{Z}\right)\right\Vert _{F}\le\delta$
with high probability. Then, instead of \eqref{eq:convex1}, the authors
in \cite{CandesPlan2009} propose solving the program 
\begin{equation}
\begin{array}{ll}
\mathrm{minimize} & \left\Vert \mathbf{X}\right\Vert _{*}\\
\mathrm{subject\, to} & \left\Vert \mathcal{P}\left(\mathbf{X}-\mathbf{Y}\right)\right\Vert _{F}\le\delta
\end{array}.\label{eq:convex1-1}
\end{equation}
Indeed, under the same hypotheses of Theorem 2, it can be shown that
the error norm $\left\Vert \mathbf{M}-\widehat{\mathbf{M}}\right\Vert _{F}$
is bounded from above as \cite{CandesPlan2009} 
\begin{equation}
\left\Vert \mathbf{M}-\widehat{\mathbf{M}}\right\Vert _{F}\le4\sqrt{\dfrac{\left(2N_{1}N_{2}+m\right)\min\left(N_{1},N_{2}\right)}{m}}\delta+2\delta
\end{equation}
with very high probability, where $\widehat{\mathbf{M}}$ denotes
the solution to \eqref{eq:convex1-1}. In a nutshell, as the authors
in \cite{CandesPlan2009} aptly explain, ``when perfect noiseless
recovery occurs, then matrix completion is \textit{stable} vis a vis
perturbations''. Consequently, if $\mathbf{M}$ exhibits favorable
coherence properties, both noiseless and noisy matrix completion will
be more realistic. Therefore, in order to guarantee satisfactory performance
for matrix completion, it suffices to study the conditions under which
the coherence of $\mathbf{M}$ will be as low as possible.

\subsection{Matrix Completion in Colocated MIMO Radar}

We consider the respective problem formulation proposed in \cite{SunPetropuluAEROSPACE}
(see Subsection A of Section II in \cite{SunPetropuluAEROSPACE}).
Extending the problem to the case of arbitrary $2$-dimensional arrays
in an obvious way \cite{KrimViberg1996}, the matrix to be completed
at the fusion center of the receiver, $\mathcal{P}\left(\mathbf{Y}\right)$,
obeys the special observation model 
\begin{equation}
\mathbf{Y}\triangleq\boldsymbol{\Delta}+\mathbf{Z}\in\mathbb{C}^{M_{r}\times M_{t}},
\end{equation}
where $M_{r}$ and $M_{t}$ denote the numbers of reception and transmission
antennas, respectively, $\mathbf{Z}$ is an interference/observation
noise matrix as defined above and 
\begin{equation}
\boldsymbol{\Delta}\triangleq\mathbf{X}_{r}\mathbf{D}\mathbf{X}_{t}^{\boldsymbol{T}},
\end{equation}
where $\mathbf{X}_{r}\in\mathbb{C}^{M_{r}\times K}$ (respectively
for $\mathbf{X}_{t}\in\mathbb{C}^{M_{t}\times K}$) constitutes an
\textit{alternant} matrix defined as 
\begin{flalign}
\mathbf{X}_{r} & \triangleq\begin{bmatrix}\gamma_{0}^{0} & \gamma_{1}^{0} & \cdots & \gamma_{K-1}^{0}\\
\gamma_{0}^{1} & \gamma_{1}^{1} & \cdots & \gamma_{K-1}^{1}\\
\vdots & \vdots & \ddots & \vdots\\
\gamma_{0}^{M_{t}-1} & \gamma_{1}^{M_{t}-1} & \cdots & \gamma_{K-1}^{M_{t}-1}
\end{bmatrix}\in\mathbb{C}^{M_{r}\times K},
\end{flalign}
with 
\begin{flalign}
\gamma_{k}^{l} & \triangleq e^{j2\pi\mathbf{r}_{r}^{\boldsymbol{T}}\left(l\right)\mathcal{T}\left(\theta_{k}\right)},\quad\left(l,k\right)\in\mathbb{N}_{M_{r}-1}\times\mathbb{N}_{K-1}\\
\mathbf{r}_{r}\left(l\right) & \triangleq\dfrac{1}{\lambda}\left[x_{l}^{r}\, y_{l}^{r}\right]^{\boldsymbol{T}}\in\mathbb{R}^{2\times1},\quad l\in\mathbb{N}_{M_{r}-1}\\
\mathcal{T}\left(\theta_{k}\right) & \triangleq\begin{bmatrix}\cos\left(\theta_{k}\right)\\
\sin\left(\theta_{k}\right)
\end{bmatrix}\in\mathbb{R}^{2\times1},\quad k\in\mathbb{N}_{K-1},
\end{flalign}
with the sets $\left\{ \left[x_{l}^{r}\, y_{l}^{r}\right]^{\boldsymbol{T}}\right\} _{l\in\mathbb{N}_{M_{r}-1}}$
and $\left\{ \theta_{k}\right\} _{k\in\mathbb{N}_{K-1}}$ containing
the $2$-dimensional antenna coordinates of the reception array and
the target angles, respectively and $\lambda$ denoting the carrier
wavelength utilized for the communication, and where $\mathbf{D}\in\mathbb{C}^{K\times K}$
constitutes a non - zero - diagonal matrix defined as 
\begin{flalign}
\mathbf{D} & \triangleq\mathrm{diag}\left(\left[\zeta_{1}\rho_{1}\,\zeta_{2}\rho_{2}\,\ldots\,\zeta_{K}\rho_{K}\right]\right),\quad\text{with}\\
\rho_{i} & \triangleq\exp\left(j\dfrac{4\pi}{\lambda}\vartheta_{i}\left(q-1\right)T_{PR}\right),\quad i\in\mathbb{N}_{K-1},
\end{flalign}
where the sets $\left\{ \zeta_{k}\right\} _{k\in\mathbb{N}_{K-1}}$
and $\left\{ \vartheta_{k}\right\} _{k\in\mathbb{N}_{K-1}}$ contain
the target reflection coefficients and target speeds, respectively
and $q$ and $T_{PR}$ denote the pulse index and the pulse repetition
interval, respectively. For the simplest ULA case (as treated in \cite{SunPetropuluAEROSPACE}),
\begin{equation}
\left[x_{l}^{r\left(t\right)}\, y_{l}^{r\left(t\right)}\right]^{\boldsymbol{T}}\equiv\left[0\, ld_{r\left(t\right)}\right]^{\boldsymbol{T}},\quad l\in\mathbb{N}_{M_{r\left(t\right)}-1},
\end{equation}
where $d_{r\left(t\right)}$ denotes the respective array antenna
spacing, and $\mathbf{X}_{r}$ and $\mathbf{X}_{t}$ degenerate to
Vandermonde matrices.

In the next sections, we focus almost exclusively on bounding the
coherence of $\boldsymbol{\Delta}$, first for the ULA case and then
for the more general arbitrary $2$-dimensional array case. Also,
when possible, we derive sufficient conditions under which the coherence
the aforementioned matrix is small.

\section{Coherence and Recoverability of $\boldsymbol{\Delta}$ for\protect
\protect \\
 ULA Transmitter - Receiver Pairs}

In this section, we present our coherence and recoverability results
for the case in which ULAs are utilized for transmission and reception.
In short, we prove: 
\begin{itemize}
\item Asymptotic and approximate optimality of the coherence of $\boldsymbol{\Delta}$
with respect to the number of transmission/reception antennas and 
\item Near optimal recoverability of $\boldsymbol{\Delta}$ via matrix completion. 
\end{itemize}

\subsection{Coherence of $\boldsymbol{\Delta}$}
\begin{thm}
\textbf{\textup{(Coherence for ULAs)}} Consider a Uniform Linear Array
(ULA) transmitter - receiver pair and assume that the set of target
angles $\left\{ \theta_{k}\right\} _{k\in\mathbb{N}_{K-1}}$ consists
of almost surely distinct members. Then, for any fixed $M_{t}$ and
$M_{r}$, as long as 
\begin{equation}
K\le\max_{i\in\left\{ t,r\right\} }\left\{ \dfrac{M_{i}}{\sqrt{\beta_{\xi_{i}}\left(M_{i}\right)}}\right\} ,
\end{equation}
the associated matrix $\boldsymbol{\Delta}$ obeys the assumptions
$\mathbf{A0}$ and $\mathbf{A1}$ with 
\begin{flalign}
\mu_{0} & \triangleq\max_{i\in\left\{ t,r\right\} }\left\{ \dfrac{M_{i}}{M_{i}-\left(K-1\right)\sqrt{\beta_{\xi_{i}}\left(M_{i}\right)}}\right\} \quad\text{and}\label{eq:mu_o}\\
\mu_{1} & \triangleq\max_{i\in\left\{ t,r\right\} }\left\{ \dfrac{M_{i}\sqrt{K}}{M_{i}-\left(K-1\right)\sqrt{\beta_{\xi_{i}}\left(M_{i}\right)}}\right\} \label{eq:mu_1}
\end{flalign}
with probability $1$. In the above, 
\begin{flalign}
\beta_{\xi_{k}}\left(M_{k}\right) & \triangleq\sup_{x\in\left[\xi_{k},\frac{1}{2}\right]}\dfrac{\sin^{2}\left(\pi M_{k}x\right)}{\sin^{2}\left(\pi x\right)},
\end{flalign}
\begin{equation}
\xi_{k}\triangleq\min_{\substack{\left(i,j\right)\in\mathbb{N}_{K-1}\times\mathbb{N}_{K-1}\\
i\neq j
}
}g\left(\dfrac{d_{k}}{\lambda}\left|\sin\left(\theta_{i}\right)-\sin\left(\theta_{j}\right)\right|\right)\label{eq:xixixixi}
\end{equation}
for $k\in\left\{ t,r\right\} $ and 
\begin{equation}
g\left(x\right)\triangleq\begin{cases}
\left\lceil x\right\rceil -x, & \left\lceil x\right\rceil -x\leq\dfrac{1}{2}\\
x-\left\lfloor x\right\rfloor , & \text{otherwise}
\end{cases}.
\end{equation}
Further, if $\xi\triangleq\min\left\{ \xi_{r},\xi_{t}\right\} \neq0$,
then, for any fixed $K$, as long as 
\begin{equation}
\min_{i\in\left\{ t,r\right\} }M_{i}\ge K\sqrt{\beta_{\xi}}=\mathcal{O}\left(K\right),
\end{equation}
where 
\begin{flalign}
\beta_{\xi} & \triangleq\sup_{\substack{x\in\left[\xi,\frac{1}{2}\right]\\
M_{t}\in\mathbb{R}_{++}
}
}\dfrac{\sin^{2}\left(\pi M_{t}x\right)}{\sin^{2}\left(\pi x\right)}\equiv\sup_{\substack{x\in\left[\xi,\frac{1}{2}\right]\\
M_{r}\in\mathbb{R}_{++}
}
}\dfrac{\sin^{2}\left(\pi M_{r}x\right)}{\sin^{2}\left(\pi x\right)},
\end{flalign}
\eqref{eq:mu_o} and \eqref{eq:mu_1} hold replacing both $\beta_{\xi_{t}}\left(M_{t}\right)$
and $\beta_{\xi_{r}}\left(M_{r}\right)$ by the constant $\beta_{\xi}$
(that is, independent of both $M_{t}$ and $M_{r}$). Additionally,
in the limit we respect to $M_{t}$ and $M_{r}$, we have 
\begin{equation}
\mu\left(V\right)\equiv\mu\left(U\right)\equiv1,
\end{equation}
that is, the coherence of $\boldsymbol{\Delta}$ is asymptotically
optimal. Finally, if $\xi$ is safely bounded away from zero, then,
for sufficiently large $M_{t}$ and $M_{r}$, it is true that 
\begin{equation}
\mu\left(V\right)\approx\mu\left(U\right)\approx1.
\end{equation}
that is, the smallest possible coherence for $\boldsymbol{\Delta}$
can be approximately attained even for finite values of $M_{t}$ and
$M_{r}$. 
\end{thm}
Before we proceed with the proof of the theorem, let us state the
following standard result, which will help us develop our arguments
easily and concretely. 
\begin{thm}
\textup{\cite{Wolkowicz1980} }Let $\mathbf{M}\in\mathbb{C}^{N\times N}$
be a matrix with real eigenvalues. Define 
\begin{equation}
\tau\triangleq\dfrac{\mathrm{tr}\left(\mathbf{M}\right)}{N}\quad\text{and}\quad s^{2}\triangleq\dfrac{\mathrm{tr}\left(\mathbf{M}^{2}\right)}{N}-\tau^{2}.
\end{equation}
Then, it is true that 
\begin{flalign}
\tau-s\sqrt{N-1} & \leq\lambda_{min}\left(\mathbf{M}\right)\leq\tau-\dfrac{s}{\sqrt{N-1}}\quad\text{and}\label{eq:bound1}\\
\tau+\dfrac{s}{\sqrt{N-1}} & \leq\lambda_{max}\left(\mathbf{M}\right)\leq\tau+s\sqrt{N-1}.\label{eq:bound2}
\end{flalign}
Further, equality holds on the left (right) of \eqref{eq:bound1}
if and only if equality holds on the left (right) of \eqref{eq:bound2}
if and only if the $N-1$ largest (smallest) eigenvalues are equal. 
\end{thm}
\begin{proof}[Proof of Theorem 3]In order to make it more tractable,
we divide the proof into the following two subsections. In the first
subsection, we present a suitable characterization of the SVD of $\boldsymbol{\Delta}$
which will come in handy in identifying the essential actions needed
to be taken in order to bound its coherence. However, the methodology
presented is very general and can be applied to any given rank-$L$
matrix $\mathbf{M}$, as long as it can be explicitly written as the
product $\mathbf{M}=\mathbf{M}_{1}\mathbf{D}\mathbf{M}_{2}\in\mathbb{C}^{M_{1}\times M_{2}}$,
with $\mathbf{D}\in\mathbb{C}^{L\times L}$ being an arbitrary full
rank matrix. In the second subsection, we use the results from the
first, as well as the special structure of $\boldsymbol{\Delta}$
in order to derive quantitatively useful results regarding its coherence.

\subsubsection{Characterization of the SVD of $\boldsymbol{\Delta}$}

We would like to study the coherence of the almost surely rank-$K$
matrix 
\begin{equation}
\boldsymbol{\Delta}=\mathbf{X}_{r}\mathbf{D}\mathbf{X}_{t}^{\boldsymbol{T}}\in\mathbb{C}^{M_{r}\times M_{t}},
\end{equation}
where the Vandermonde matrix $\mathbf{X}_{r}\in\mathbb{C}^{M_{r}\times K}$
(respectively for $\mathbf{X}_{t}\in\mathbb{C}^{M_{t}\times K}$)
is defined by the generating vector 
\begin{equation}
\Gamma\triangleq\left[\gamma_{0}\,\gamma_{1}\,\ldots\,\gamma_{K-1}\right]^{\boldsymbol{T}}\in\mathbb{C}^{K\times1},
\end{equation}
with 
\begin{equation}
\gamma_{k}\triangleq e^{j2\pi\alpha_{k}^{r}}\quad\text{and}\quad\alpha_{k}^{r}\triangleq\frac{d_{r}\sin\left(\theta_{k}\right)}{\lambda},\quad k\in\mathbb{N}_{K-1},
\end{equation}
and $\mathbf{D}\in\mathbb{C}^{K\times K}$ is an arbitrary non - zero
- diagonal matrix. Since $\boldsymbol{\Delta}$ is assumed to be of
rank-$K$ almost surely, we implicitly consider the case where the
set of angles $\left\{ \theta_{k}\right\} _{k\in\mathbb{N}_{K-1}}$
consists of almost surely distinct members.

In general, the compact SVD of $\boldsymbol{\Delta}$ can be expressed
as 
\begin{equation}
\boldsymbol{\Delta}=\mathbf{U}\boldsymbol{\Sigma}\mathbf{V}^{\boldsymbol{H}},
\end{equation}
where $\mathbf{U}\in\mathbb{C}^{M_{r}\times K}$, $\mathbf{V}\in\mathbb{C}^{M_{t}\times K}$
such that $\mathbf{U}^{\boldsymbol{H}}\mathbf{U}=\mathbf{I}_{K}$,
$\mathbf{V}^{\boldsymbol{H}}\mathbf{V}=\mathbf{I}_{K}$ and $\boldsymbol{\Sigma}\in\mathbb{R}^{K\times K}$
constitutes a diagonal matrix containing the (non zero) singular values
of $\boldsymbol{\Delta}$.

Respectively for $\mathbf{X}_{t}$, consider the thin QR decomposition
of $\mathbf{X}_{r}$ given by $\mathbf{X}_{r}=\mathbf{V}_{r}\mathbf{A}_{r}$,
where $\mathbf{V}_{r}\in\mathbb{C}^{M_{r}\times K}$ is such that
$\mathbf{V}_{r}^{\boldsymbol{\mathit{H}}}\mathbf{V}_{r}\equiv\mathbf{I}_{K}$
and $\mathbf{A}_{r}\in\mathbb{C}^{K\times K}$ constitutes an upper
triangular matrix. Then, $\boldsymbol{\Delta}=\mathbf{V}_{r}\mathbf{A}_{r}\mathbf{D}\mathbf{A}_{t}^{\boldsymbol{\mathit{T}}}\mathbf{V}_{t}^{\boldsymbol{\mathit{T}}}$
and since the matrix $\mathbf{A}_{r}\mathbf{D}\mathbf{A}_{t}^{\boldsymbol{\mathit{T}}}\in\mathbb{C}^{K\times K}$
is almost surely of full rank by definition, its SVD is given by $\mathbf{A}_{r}\mathbf{D}\mathbf{A}_{t}^{\boldsymbol{\mathit{T}}}=\mathbf{Q}_{r}\boldsymbol{\Lambda}\mathbf{Q}_{t}^{\boldsymbol{\mathit{H}}}$,
where $\mathbf{Q}_{r}\in\mathbb{C}^{K\times K}$ is such that $\mathbf{Q}_{r}\mathbf{Q}_{r}^{\boldsymbol{\mathit{H}}}=\mathbf{Q}_{r}^{\boldsymbol{\mathit{H}}}\mathbf{Q}_{r}\equiv\mathbf{I}_{K}$
(respectively for $\mathbf{Q}_{r}$) and $\boldsymbol{\Lambda}\in\mathbb{R}^{K\times K}$
is non - zero - diagonal, containing the singular values of $\mathbf{A}_{r}\mathbf{D}\mathbf{A}_{t}^{\boldsymbol{\mathit{T}}}$.
Thus, we arrive at the expression 
\begin{equation}
\boldsymbol{\Delta}=\mathbf{V}_{r}\mathbf{Q}_{r}\boldsymbol{\Lambda}\mathbf{Q}_{t}^{\boldsymbol{\mathit{H}}}\mathbf{V}_{t}^{\boldsymbol{\mathit{T}}}\equiv\mathbf{V}_{r}\mathbf{Q}_{r}\boldsymbol{\Lambda}\left(\mathbf{V}_{t}^{*}\mathbf{Q}_{t}\right)^{\boldsymbol{\mathit{H}}},
\end{equation}
which constitutes a valid SVD of $\boldsymbol{\Delta}$, since $\left(\mathbf{V}_{r}\mathbf{Q}_{r}\right)^{\boldsymbol{\mathit{H}}}\mathbf{V}_{r}\mathbf{Q}_{r}\equiv\mathbf{I}_{K}$
and $\left(\mathbf{V}_{t}^{*}\mathbf{Q}_{t}\right)^{\boldsymbol{\mathit{H}}}\mathbf{V}_{t}^{*}\mathbf{Q}_{t}\equiv\mathbf{I}_{K}$
and consequently, by the uniqueness of the singular values of a matrix,
$\boldsymbol{\Lambda}\equiv\boldsymbol{\Sigma}$. Therefore, we can
set $\mathbf{U}=\mathbf{V}_{r}\mathbf{Q}_{r}$ and $\mathbf{V}=\mathbf{V}_{t}^{*}\mathbf{Q}_{t}$.

If $\mathbf{V}_{r}^{i}\in\mathbb{C}^{1\times K}$ and $\mathbf{X}_{r}^{i}\in\mathbb{C}^{1\times K},i\in\mathbb{N}_{M_{r}}^{+}$
denote the $i$-th row of $\mathbf{V}_{r}$ and $\mathbf{X}_{r}$,
respectively, it holds that 
\begin{flalign}
\mu\left(U\right) & =\frac{M_{r}}{K}\sup_{i\in\mathbb{N}_{N}^{+}}\left\Vert \mathbf{V}_{r}^{i}\mathbf{Q}_{r}\right\Vert _{2}^{2}\nonumber \\
 & \equiv\frac{M_{r}}{K}\sup_{i\in\mathbb{N}_{N}^{+}}\left\Vert \mathbf{V}_{r}^{i}\right\Vert _{2}^{2}\label{eq:RECALL}\\
 & =\frac{M_{r}}{K}\sup_{i\in\mathbb{N}_{N}^{+}}\left\Vert \mathbf{X}_{r}^{i}\mathbf{A}_{r}^{-1}\right\Vert _{2}^{2}\label{eq:almostsurely-2}\\
 & \leq\frac{M_{r}}{K}\sup_{i\in\mathbb{N}_{N}^{+}}\dfrac{\left\Vert \mathbf{X}_{r}^{i}\right\Vert _{2}^{2}}{\sigma_{min}^{2}\left(\mathbf{A}_{r}\right)}\nonumber \\
 & =\dfrac{M_{r}}{\sigma_{min}^{2}\left(\mathbf{A}_{r}\right)}\nonumber \\
 & \equiv\dfrac{M_{r}}{\lambda_{min}\left(\mathbf{A}_{r}^{\boldsymbol{\mathit{H}}}\mathbf{A}_{r}\right)}\nonumber \\
 & \equiv\dfrac{M_{r}}{\lambda_{min}\left(\mathbf{A}_{r}^{\boldsymbol{\mathit{H}}}\mathbf{V}_{r}^{\boldsymbol{\mathit{H}}}\mathbf{V}_{r}\mathbf{A}_{r}\right)}
\end{flalign}
or, equivalently, 
\begin{equation}
\mu\left(U\right)\leq\dfrac{M_{r}}{\lambda_{min}\left(\mathbf{X}_{r}^{\boldsymbol{\mathit{H}}}\mathbf{X}_{r}\right)},\label{eq:first_bound}
\end{equation}
where \eqref{eq:almostsurely-2} stems from the fact that the columns
of $\mathbf{X}_{r}$ are linearly independent almost surely and, as
a result, $\mathbf{A}_{r}$ is almost surely invertible. Likewise,
regarding the coherence of the row space of $\boldsymbol{\Delta}$,
we get 
\begin{equation}
\mu\left(V\right)\leq\dfrac{M_{t}}{\lambda_{min}\left(\mathbf{X}_{t}^{\boldsymbol{\mathit{H}}}\mathbf{X}_{t}\right)}.\label{eq:second_bound}
\end{equation}

\subsubsection{Effective Bounding of $\lambda_{min}\left(\mathbf{X}_{t\left(r\right)}^{\boldsymbol{\mathit{H}}}\mathbf{X}_{t\left(r\right)}\right)$}

In the following, we will mainly focus on the matrix $\mathbf{X}_{t}^{\boldsymbol{\mathit{H}}}\mathbf{X}_{t}$,
since the respective analysis for $\mathbf{X}_{r}^{\boldsymbol{\mathit{H}}}\mathbf{X}_{r}$
is identical.

Due to \eqref{eq:first_bound}, we need a strictly positive lower
bound for $\lambda_{min}\left(\mathbf{X}_{t}^{\boldsymbol{\mathit{H}}}\mathbf{X}_{t}\right)$,
with 
\begin{equation}
\mathbf{X}_{t}^{\boldsymbol{\mathit{H}}}\mathbf{X}_{t}\triangleq\begin{bmatrix}M_{t} & \delta_{1,0} & \cdots & \delta_{K-1,0}\\
\delta_{1,0}^{*} & M_{t} & \cdots & \delta_{K-1,1}\\
\vdots & \vdots & \ddots & \vdots\\
\delta_{K-1,0}^{*} & \delta_{K-1,1}^{*} & \cdots & M_{t}
\end{bmatrix},\label{eq:Vandermonde}
\end{equation}
where 
\begin{gather}
\delta_{i,j}\triangleq{\displaystyle \sum_{m=0}^{M_{t}-1}}e^{j2\pi m\left(\alpha_{i}^{t}-\alpha_{j}^{t}\right)},\;\forall\,\left(i,j\right)\in\mathbb{N}_{K-1}\times\mathbb{N}_{K-1}.
\end{gather}
Obviously, $M_{t}\equiv\delta_{i,i},\forall\, i\in\mathbb{N}_{K-1}.$

Let us now apply Theorem 4 to the almost surely positive definite
matrix $\mathbf{M}\triangleq\mathbf{X}_{t}^{\boldsymbol{\mathit{H}}}\mathbf{X}_{t}\in\mathbb{C}^{K\times K}$.
The trace of $\mathbf{M}$ is simply $M_{t}K$. Hence, 
\begin{equation}
\tau=\dfrac{M_{t}K}{K}\equiv M_{t}.
\end{equation}
\begin{figure}
\centering\includegraphics[scale=0.6]{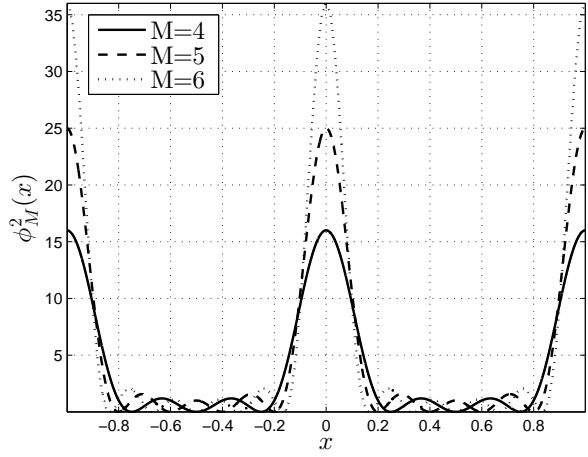}\caption{\label{fig:MagnitudeSinc}Square of $\phi_{M}\left(x\right)$ for
different values of $M$.}
\end{figure}

\noindent We also need the trace of $\mathbf{M}^{2}$. Since $\mathbf{M}$
is a Hermitian matrix, it is true that 
\begin{flalign}
\mathrm{tr}\left(\mathbf{M}^{2}\right) & =\sum_{k_{1}=0}^{K-1}\sum_{k_{2}=0}^{K-1}\left|\delta_{k_{1},k_{2}}\right|^{2}\nonumber \\
 & \equiv\sum_{k_{1}=0}^{K-1}\left\{ M_{t}^{2}+\sum_{\substack{k_{2}=0\\
k_{1}\neq k_{2}
}
}^{K-1}\left|{\displaystyle \sum_{m=0}^{M_{t}-1}}e^{j2\pi m\left(\alpha_{k_{1}}^{t}-\alpha_{k_{2}}^{t}\right)}\right|^{2}\right\} \nonumber \\
 & =\sum_{k_{1}=0}^{K-1}\left\{ M_{t}^{2}+\sum_{\substack{k_{2}=0\\
k_{1}\neq k_{2}
}
}^{K-1}\dfrac{\sin^{2}\left(\pi M_{t}\left(\alpha_{k_{1}}^{t}-\alpha_{k_{2}}^{t}\right)\right)}{\sin^{2}\left(\pi\left(\alpha_{k_{1}}^{t}-\alpha_{k_{2}}^{t}\right)\right)}\right\} \nonumber \\
 & \triangleq\sum_{k_{1}=0}^{K-1}\left\{ M_{t}^{2}+\sum_{\substack{k_{2}=0\\
k_{1}\neq k_{2}
}
}^{K-1}\phi_{M_{t}}^{2}\left(\alpha_{k_{1}}^{t}-\alpha_{k_{2}}^{t}\right)\right\} .\label{eq:trace_square}
\end{flalign}
At this point, it is instructive to at least qualitatively study the
behavior of the square of $\phi_{M}\left(x\right)$ (for a general
parameter $M\in\mathbb{N}^{+}$) defined above%
\footnote{This function constitutes a trivial variation of the so called Dirichlet
kernel or periodic sinc function.%
}. Fig. \ref{fig:MagnitudeSinc} shows the square of the function for
three different values of $M$. We can directly identify the following
useful fact:
\begin{fact}
If 
\begin{equation}
x\in\left[k,k+\dfrac{1}{2}\right],\:\forall\, k\in\mathbb{Z},
\end{equation}
then the entire sequence of local maxima of $\phi_{M}^{2}\left(x\right)$
is strictly decreasing. 
\end{fact}
\noindent Next, define the function $g:\mathbb{R}_{+}\rightarrow\mathbb{R}_{+}$
as 
\begin{equation}
g\left(x\right)\triangleq\begin{cases}
\left\lceil x\right\rceil -x, & \left\lceil x\right\rceil -x\leq\dfrac{1}{2}\\
x-\left\lfloor x\right\rfloor , & \text{otherwise}
\end{cases}
\end{equation}
and let 
\begin{equation}
\xi_{t}\triangleq\min_{\substack{\left(i,j\right)\in\mathbb{N}_{K-1}\times\mathbb{N}_{K-1}\\
i\neq j
}
}g\left(\left|\alpha_{i}^{t}-\alpha_{j}^{t}\right|\right)\in\left[0,\dfrac{1}{2}\right].
\end{equation}
The motivation for defining $\xi_{t}$ stems from the fact that the
function $\phi_{M}^{2}\left(x\right)$ is both periodic with unitary
period and symmetric and thus it suffices to study its behavior only
for $x\in\left[0,\dfrac{1}{2}\right]$. Then, due to Fact 1, we can
upper bound \eqref{eq:trace_square} as 
\begin{flalign}
\mathrm{tr}\left(\mathbf{M}^{2}\right) & \leq\sum_{k_{1}=0}^{K-1}\left\{ M_{t}^{2}+\left(K-1\right)\sup_{x\in\left[\xi_{t},\frac{1}{2}\right]}\phi_{M_{t}}^{2}\left(x\right)\right\} \\
 & \triangleq KM_{t}^{2}+K\left(K-1\right)\beta_{\xi_{t}}\left(M_{t}\right).
\end{flalign}
Consequently, by Theorem 4, 
\begin{flalign}
s^{2} & =\dfrac{\mathrm{tr}\left(\mathbf{M}^{2}\right)}{K}-M_{t}^{2}\nonumber \\
 & \leq M_{t}^{2}+\beta_{\xi}\left(M_{t}\right)\left(K-1\right)-M_{t}^{2}\nonumber \\
 & \equiv\beta_{\xi_{t}}\left(M_{t}\right)\left(K-1\right)
\end{flalign}
and consequently we can bound $\lambda_{min}\left(\mathbf{M}\right)\equiv\lambda_{min}\left(\mathbf{X}_{t}^{\boldsymbol{\mathit{H}}}\mathbf{X}_{t}\right)$
from below as 
\begin{equation}
\lambda_{min}\left(\mathbf{X}_{t}^{\boldsymbol{\mathit{H}}}\mathbf{X}_{t}\right)\geq M_{t}-\left(K-1\right)\sqrt{\beta_{\xi_{t}}\left(M_{t}\right)}.
\end{equation}
Respectively, we get 
\begin{equation}
\lambda_{min}\left(\mathbf{X}_{r}^{\boldsymbol{\mathit{H}}}\mathbf{X}_{r}\right)\geq M_{r}-\left(K-1\right)\sqrt{\beta_{\xi_{r}}\left(M_{r}\right)}.
\end{equation}
Therefore, for fixed $M_{t}$ and $M_{r}$, as long as 
\begin{equation}
K\le\max_{i\in\left\{ t,r\right\} }\left\{ \dfrac{M_{i}}{\sqrt{\beta_{\xi_{i}}\left(M_{i}\right)}}\right\} ,\label{eq:bound_K}
\end{equation}
the upper bounds 
\begin{flalign}
\mu\left(U\right) & \leq\dfrac{M_{t}}{M_{t}-\left(K-1\right)\sqrt{\beta_{\xi_{t}}\left(M_{t}\right)}}\quad\text{and}\\
\mu\left(V\right) & \leq\dfrac{M_{r}}{M_{r}-\left(K-1\right)\sqrt{\beta_{\xi_{r}}\left(M_{r}\right)}}
\end{flalign}
both hold true with probability 1. Consequently, 
\begin{equation}
\max\left\{ \mu\left(U\right),\mu\left(V\right)\right\} \leq\max_{i\in\left\{ t,r\right\} }\left\{ \dfrac{M_{i}}{M_{i}-\left(K-1\right)\sqrt{\beta_{\xi_{i}}\left(M_{i}\right)}}\right\} \triangleq\mu_{0}.
\end{equation}
Fig. \ref{fig:sup_over_xi} depicts the parameter $\sqrt{\beta_{\xi}\left(M\right)}$
(either for $\xi_{r}$ or $\xi_{t}$ and $M_{r}$ or $M_{t}$) as
a function of $\xi$ for three values of $M$. 
\begin{figure}
\centering\includegraphics[scale=0.6]{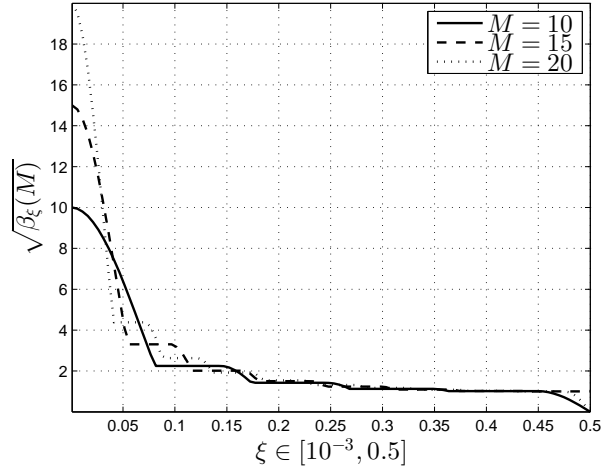}\caption{\label{fig:sup_over_xi}Values of $\sqrt{\beta_{\xi}\left(M\right)}$
for different values of $M$.}
\end{figure}

Further, the following fact also holds true, stemming directly from
the behavior of $\phi_{M}^{2}\left(x\right)$ (also see Fig. \ref{fig:MagnitudeSinc}).
\begin{figure}
\centering\includegraphics[scale=0.6]{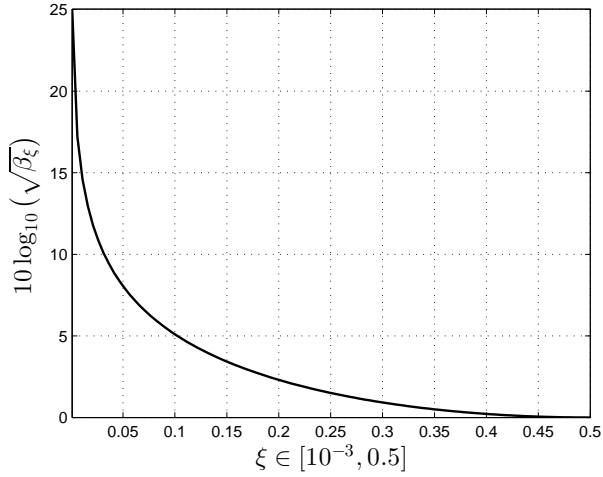}\caption{\label{fig:sup_over_all}Values of $\sqrt{\beta_{\xi}}$ in the interval
$\left[10^{-3},0.5\right]$.}
\end{figure}

\begin{fact}
Restricting our attention at the principal period of $\phi_{M}^{2}\left(x\right)$
($\left[-1,1\right]$), as M is getting larger, the main lobe of the
function gets taller and narrower. Also, by definition of $\phi_{M}^{2}\left(x\right)$,
it follows that for any $M\in\mathbb{N}^{+}$ (and even when $M\rightarrow\infty$),
the function possesses only one singularity (in the principal period)
occurring at $x\equiv0$. 
\end{fact}
Now, assume that $\xi\triangleq\min\left\{ \xi_{r},\xi_{t}\right\} \neq0$.
Then, we can write 
\begin{flalign}
\beta_{\xi_{t}}\left(M_{t}\right) & \equiv\sup_{x\in\left[\xi_{t},\frac{1}{2}\right]}\phi_{M_{t}}^{2}\left(x\right)\nonumber \\
 & \le\sup_{x\in\left[\xi,\frac{1}{2}\right]}\phi_{M_{t}}^{2}\left(x\right)\nonumber \\
 & \leq\sup_{\substack{x\in\left[\xi,\frac{1}{2}\right]\\
M_{t}\in\mathbb{N}^{+}
}
}\phi_{M_{t}}^{2}\left(x\right)\equiv\sup_{\substack{x\in\left[\xi,\frac{1}{2}\right]\\
M_{r}\in\mathbb{N}^{+}
}
}\phi_{M_{r}}^{2}\left(x\right)\triangleq\beta_{\xi}.\label{eq:sup_over_all}
\end{flalign}
Since $\xi\neq0$ by assumption and due to Fact 2, the supremum appearing
in \eqref{eq:sup_over_all} will be always finite and, consequently
the respective bound will be non trivial.

Therefore, in this case, for fixed $K$, \eqref{eq:bound_K} becomes
\begin{equation}
\min_{i\in\left\{ t,r\right\} }M_{i}\ge K\sqrt{\beta_{\xi}}=\mathcal{O}\left(K\right).
\end{equation}
Passing to the limit with respect to $M_{t}$ (respectively $M_{r}$),
we get 
\begin{flalign}
\limsup_{M_{t}\rightarrow\infty}\mu\left(U\right) & \leq\limsup_{M_{t}\rightarrow\infty}\dfrac{M_{t}}{M_{t}-\left(K-1\right)\sqrt{\beta_{\xi}}}\nonumber \\
 & \equiv\lim_{M_{t}\rightarrow\infty}\dfrac{M_{t}}{M_{t}-\left(K-1\right)\sqrt{\beta_{\xi}}}=1
\end{flalign}
and since $\mu\left(U\right)\geq1$ by definition \cite{Candes&Recht2009},
it must be true that, \textit{in the limit}, $\mu\left(U\right)\equiv1$.
Likewise, it will also be true that, \textit{in the limit}, $\mu\left(V\right)\equiv1$.

If additionally $\xi$ is safely bounded away from zero, then $\beta_{\xi}$
will be small (due to Fact 1) and thus for sufficiently large but
finite $M_{t}$ and $M_{r}$, 
\begin{equation}
\mu\left(U\right)\approx\mu\left(V\right)\approx1.
\end{equation}
An experimental calculation of the square root of the supremum $\beta_{\xi}$
for $\xi\in\left[10^{-3},0.5\right]$ shown in Fig. \ref{fig:sup_over_all}
(in logarithmic scale) can validate and possibly strengthen our theoretical
arguments presented above. We observe that $\sqrt{\beta_{\xi}}$ over
$\xi$ constitutes a hyperbolic function, converging rapidly to 1
(linear scale), which corresponds to its minimum value.

Finally, regarding the Assumption $\mathbf{A1}$, by a simple argument
involving the Cauchy - Schwarz inequality, it can be shown that \cite{Candes&Recht2009},
in the general case, one can choose 
\begin{equation}
\mu_{1}\triangleq\mu_{0}\sqrt{K}=\max_{i\in\left\{ t,r\right\} }\left\{ \dfrac{M_{i}\sqrt{K}}{M_{i}-\left(K-1\right)\sqrt{\beta_{\xi_{i}}\left(M_{i}\right)}}\right\} ,
\end{equation}
holding true also with probability 1. The results presented in the
statement of Theorem 3 readily follow.\end{proof}

In colocated MIMO Radar systems, due to the need for unambiguous angle
estimation (target detection), it is very common to assume that $\theta_{i}\in\left[-\pi/2,\pi/2\right],\forall\, i\in\mathbb{N}_{K-1}$.
Also, especially for the case where ULAs are employed for transmission
and reception, another common assumption is to choose $d_{r}\equiv d_{t}=\lambda/2$.
Under this setting, the following lemma provides a simple sufficient
condition, which guarantees that the asymptotics of Theorem 3 hold
true and, as a result, that for a sufficiently large number of transmission/reception
antennas, the coherence of $\boldsymbol{\Delta}$ will be small. 
\begin{lem}
\textbf{\textup{(ULA Pairs Coherence Control)}} Consider the hypotheses
and definitions of Theorem 3. Set $d_{r}\equiv d_{t}=\lambda/2$.
If additionally 
\begin{gather}
\left(\theta_{i},\theta_{j}\right)\in\mathcal{A},\quad\text{with}\\
\mathcal{A}\triangleq\left\{ \left.\left(x,y\right)\in\left[-\dfrac{\pi}{2},\dfrac{\pi}{2}\right]^{2}\right|\eta\leq\left|y-x\right|\leq\pi-\eta\right\} \label{eq:SET_LEMMA2}
\end{gather}
$\forall\,\left(i,j\right)\in\mathbb{N}_{K-1}\times\mathbb{N}_{K-1}$
with $i\neq j$ and for some apriori chosen $\eta\in\left(0,\dfrac{\pi}{2}\right]$,
then 
\begin{equation}
\xi=1-\cos\left(\dfrac{\eta}{2}\right)\in\left(0,\dfrac{2-\sqrt{2}}{2}\right]\label{eq:xi_worstcase}
\end{equation}
and the asymptotics of Theorem 3 always hold true. In particular,
the higher the value of $\eta$, the higher the value of $\xi$ and
the lower the coherence of $\boldsymbol{\Delta}$. 
\end{lem}
\begin{proof}[Proof of Lemma 2]See Appendix B.\end{proof}

\subsection{Near Optimal Recoverability of $\boldsymbol{\Delta}$}

Using Lemma 1, we can now directly combine Theorem 2 with the coherence
results presented in the previous subsection, producing (asymptotic)
bounds regarding the number of observations required for the exact
reconstruction of the matrix $\boldsymbol{\Delta}$ by solving the
convex program \eqref{eq:convex1}. Specifically for ULA transmitter/receiver
pairs, we present the following interesting result. Since the proof
is straightforward, it is omitted. 
\begin{thm}
\textbf{\textup{(Near Optimal Recoverability for ULA Pairs)}} Consider
the hypotheses of Theorem 3 and assume that $\xi$ is safely bounded
away from zero. Also, set $M\triangleq\max\left\{ M_{r},M_{t}\right\} $.
Suppose we observe $m$ entries of $\boldsymbol{\Delta}$ with locations
sampled uniformly at random. Then, for $M_{r}$ and $M_{t}$ sufficiently
large satisfying 
\begin{equation}
\min_{i\in\left\{ t,r\right\} }M_{i}\ge K\sqrt{\beta_{\xi}}
\end{equation}
for a fixed number of targets $K$, $\boldsymbol{\Delta}$ is strongly
incoherent with parameter 
\begin{equation}
\mu\le\mu_{0}\sqrt{K}\approx\sqrt{K}
\end{equation}
and there exist positive numerical constants $C_{1}$ and $C_{2}$
such that if 
\begin{flalign}
m & \gtrsim C_{1}K^{4}M\log^{2}M=\mathcal{O}\left(M\log^{2}M\right)\quad\text{or}\\
m & \gtrsim C_{2}K^{2}M\log^{6}M=\mathcal{O}\left(M\log^{6}M\right),
\end{flalign}
the minimizer to the program \eqref{eq:convex1} is unique and equal
to $\boldsymbol{\Delta}$ with probability at least $1-M^{-3}$. 
\end{thm}
Roughly speaking, Theorem 5 implies that as long as ULAs are concerned,
for a sufficiently large number of transmission and reception antennas
and for a fixed and relatively small number of targets, matrix completion
is exact if the number of observations is at least of an order of
$M\mathrm{polylog}\left(M\right)$, that is, matrix completion is
exact for a slightly larger number of observations than the information
theoretic limit (see \cite{Candes&Tao2010} for details).

\section{Coherence of $\boldsymbol{\Delta}$ for Arbitrary 2D Arrays}

It is possible to generalize Theorem 3 for a far more general case,
that is, when the transmitter and the receiver employ arbitrary $2$-dimensional
arrays. However, in this case, since our knowledge about the specific
characteristics of the topologies of the arrays involved is very restricted,
our bounds, although tight, are not expected to be as easily handleable
as in the ULA case (see Theorem 3 and Lemma 2). Next, we state and
prove the following theorem in this respect. 
\begin{thm}
\textbf{\textup{(Coherence for Arbitrary 2D Arrays)}} Consider an
arbitrary array transmitter - receiver pair equipped with $M_{t}$
and $M_{r}$ antennas, respectively. Assume that the set of target
angles $\left\{ \theta_{k}\right\} _{k\in\mathbb{N}_{K-1}}$ consists
of almost surely distinct members and that 
\begin{equation}
\left(\theta_{i},\theta_{j}\right)\in\mathcal{A}\subseteq\mathbb{R}^{2}-\left\{ \left.\left(x,y\right)\in\mathbb{R}^{2}\right|x\neq y\right\} 
\end{equation}
$\forall\,\left(i,j\right)\in\mathbb{N}_{K-1}\times\mathbb{N}_{K-1}$
with $i\neq j$, where $\mathcal{A}$ constitutes a nominal point
set for all admissible angle pair combinations. Also, let the abstract
sets $\mathfrak{T}$ and $\mathfrak{R}$ contain all the essential
information regarding the transmitter and receiver topologies, also
assumed fixed and known apriori. Then, for any $M_{t}$ and $M_{r}$,
as long as 
\begin{equation}
K\le\max_{i\in\left\{ t,r\right\} }\left\{ \dfrac{M_{i}}{\sqrt{\beta_{i}}}\right\} ,
\end{equation}
the associated matrix $\boldsymbol{\Delta}$ obeys the assumptions
$\mathbf{A0}$ and $\mathbf{A1}$ with 
\begin{flalign}
\mu_{0} & \triangleq\max_{i\in\left\{ t,r\right\} }\left\{ \dfrac{M_{i}}{M_{i}-\left(K-1\right)\sqrt{\beta_{i}}}\right\} \quad\text{and}\label{eq:mu_o-1}\\
\mu_{1} & \triangleq\max_{i\in\left\{ t,r\right\} }\left\{ \dfrac{M_{i}\sqrt{K}}{M_{i}-\left(K-1\right)\sqrt{\beta_{i}}}\right\} \label{eq:mu_1-1}
\end{flalign}
with probability $1$ and where 
\begin{flalign}
\beta_{t\left(r\right)} & \triangleq\sup_{\left(x,y\right)\in\mathcal{A}}\left|\varphi_{t\left(r\right)}\left(x,y\left|\mathfrak{T}\left(\mathfrak{R}\right)\right.\right)\right|^{2}\in\left[0,M_{t\left(r\right)}^{2}\right),\label{eq:strange_supremum}
\end{flalign}
with $\varphi_{t\left(r\right)}:\mathcal{A}\rightarrow\mathbb{C}$
given by 
\begin{equation}
\varphi_{t\left(r\right)}\left(x,y\left|\mathfrak{T}\left(\mathfrak{R}\right)\right.\right)\triangleq{\displaystyle \sum_{m=0}^{M_{t\left(r\right)}-1}}\exp\left(j2\pi\mathbf{r}_{t\left(r\right)}^{\boldsymbol{T}}\left(m\right)\left(\mathcal{T}\left(x\right)-\mathcal{T}\left(y\right)\right)\right),
\end{equation}
where 
\begin{flalign}
\mathbf{r}_{t\left(r\right)}\left(l\right) & \triangleq\dfrac{1}{\lambda}\left[x_{l}^{t\left(r\right)}\, y_{l}^{t\left(r\right)}\right]^{\boldsymbol{T}}\in\mathbb{R}^{2\times1},l\in\mathbb{N}_{M_{t\left(r\right)}-1}\\
\mathcal{T}\left(x\right) & \triangleq\begin{bmatrix}\cos\left(x\right)\\
\sin\left(x\right)
\end{bmatrix}\in\mathbb{R}^{2\times1},
\end{flalign}
with the collection of vectors $\left\{ \left[x_{l}^{t\left(r\right)}\, y_{l}^{t\left(r\right)}\right]^{\boldsymbol{T}}\right\} _{l\in\mathbb{N}_{M_{t\left(r\right)}-1}}$
denoting the $2$-dimensional antenna coordinates of the respective
array. 
\end{thm}
\begin{proof}[Proof of Theorem 6]Again, we divide the proof into
the following subsections, in similar fashion as in the proof of Theorem
3.

\subsubsection{Characterization of the SVD of $\boldsymbol{\Delta}$}

We now consider the almost surely rank-$K$ matrix 
\begin{equation}
\boldsymbol{\Delta}=\mathbf{X}_{r}\mathbf{D}\mathbf{X}_{t}^{\boldsymbol{T}}\in\mathbb{C}^{M_{t}\times M_{r}},
\end{equation}
exactly as defined in Subsection C of Section II. Of course, since
$\boldsymbol{\Delta}$ is assumed to be of rank-$K$ almost surely,
the members of the set of angles $\left\{ \theta_{k}\right\} _{k\in\mathbb{N}_{K-1}}$
are implicitly assumed to be almost surely distinct.

By almost identical reasoning as in the proof of Theorem 3, the coherences
of the column and row spaces of $\boldsymbol{\Delta}$, $U$ and $V$,
can be upper bounded using \eqref{eq:first_bound} and \eqref{eq:second_bound},
respectively.

\subsubsection{Bounding $\lambda_{min}\left(\mathbf{X}_{t\left(r\right)}^{\boldsymbol{\mathit{H}}}\mathbf{X}_{t\left(r\right)}\right)$}

We directly focus on the matrix $\mathbf{X}_{t}^{\boldsymbol{\mathit{H}}}\mathbf{X}_{t}\in\mathbb{C}^{K\times K}$
defined as in \eqref{eq:Vandermonde}, where this time, $\forall\,\left(i,j\right)\in\mathbb{N}_{K-1}\times\mathbb{N}_{K-1}$,
\begin{gather}
\delta_{i,j}\triangleq{\displaystyle \sum_{m=0}^{M_{t}-1}}e^{j2\pi\mathbf{r}_{t}^{\boldsymbol{T}}\left(m\right)\left(\mathcal{T}\left(\theta_{i}\right)-\mathcal{T}\left(\theta_{j}\right)\right)},\label{eq:aha}
\end{gather}
with $M_{t}\equiv\delta_{i,i},\forall\, i\in\mathbb{N}_{K-1}$.

Applying Theorem 3 to $\mathbf{M}\triangleq\mathbf{X}_{t}^{\boldsymbol{\mathit{H}}}\mathbf{X}_{t}\in\mathbb{C}^{K\times K}$,
we get $\tau\equiv M_{t}$ and, concerning the trace of $\mathbf{M}^{2}$,
it is true that
\begin{flalign*}
\mathrm{tr}\left(\mathbf{M}^{2}\right) & =\sum_{k_{1}=0}^{K-1}\sum_{k_{2}=0}^{K-1}\left|\delta_{k_{1},k_{2}}\right|^{2}
\end{flalign*}
\[
=\sum_{k_{1}=0}^{K-1}\left\{ M_{t}^{2}+\sum_{\substack{k_{2}=0\\
k_{1}\neq k_{2}
}
}^{K-1}\left|{\displaystyle \sum_{m=0}^{M_{t}-1}}e^{j2\pi\mathbf{r}_{t}^{\boldsymbol{T}}\left(m\right)\left(\mathcal{T}\left(\theta_{k_{1}}\right)-\mathcal{T}\left(\theta_{k_{2}}\right)\right)}\right|^{2}\right\} .
\]
Let all the assumptions of the statement of Theorem 6 hold true. Then,
we can define a bivariate function $\varphi_{t}:\mathcal{A}\rightarrow\mathbb{C}$
given by%
\footnote{The form of this function is identical with \eqref{eq:aha}. However,
it is defined on a different set.%
} 
\begin{equation}
\varphi_{t}\left(x,y\left|\mathfrak{T}\right.\right)\triangleq{\displaystyle \sum_{m=0}^{M_{t}-1}}e^{j2\pi\mathbf{r}_{t}^{\boldsymbol{T}}\left(m\right)\left(\mathcal{T}\left(x\right)-\mathcal{T}\left(y\right)\right)},
\end{equation}
whose norm can be bounded as 
\begin{flalign}
\left|\varphi_{t}\left(x,y\left|\mathfrak{T}\right.\right)\right|^{2} & \leq\sup_{\left(x,y\right)\in\mathcal{A}}\left|\varphi_{t}\left(x,y\left|\mathfrak{T}\right.\right)\right|^{2}\in\left[0,M_{t}^{2}\right),
\end{flalign}
Observe that in the expressions above, the supremum is taken conditional
on the set $\mathfrak{T}$ and, as a result, \textit{it is highly
dependent on the array topology of the transmitter and therefore also
on} $M_{t}$, but of course independent of any combination of the
angles.

Then, we can bound $\mathrm{tr}\left(\mathbf{M}^{2}\right)$ as 
\begin{flalign}
\mathrm{tr}\left(\mathbf{M}^{2}\right) & \leq\sum_{k_{1}=0}^{K-1}\left\{ M_{t}^{2}+\left(K-1\right)\sup_{\left(x,y\right)\in\mathcal{A}}\left|\varphi_{t}\left(x,y\left|\mathfrak{T}\right.\right)\right|^{2}\right\} \nonumber \\
 & \triangleq KM_{t}^{2}+K\left(K-1\right)\beta_{t}
\end{flalign}
and the results follow using similar procedure to the respective part
of the proof of Theorem 3.\end{proof} 
\begin{rem}
Indeed, one may claim that Theorem 6 is mostly of theoretical value
and that its practical applicability is limited. On the one hand,
this is true, since, except for the ULA case, for which explicit results
are available (see Theorem 3 and Lemma 2), the computation of closed
form expressions for the supremum appearing in \eqref{eq:strange_supremum}
is almost impossible. Actually, bounding norms of sums of complex
exponentials, whose exponents are arbitrary real variable functions
constitutes a difficult mathematical problem. However, from an engineering
point of view, for a given pair of transmitter - receiver topologies,
we can always compute the aforementioned supremum empirically, as
it becomes clear by the following example. 
\begin{figure*}
\centering\subfloat[]{\includegraphics[scale=0.6]{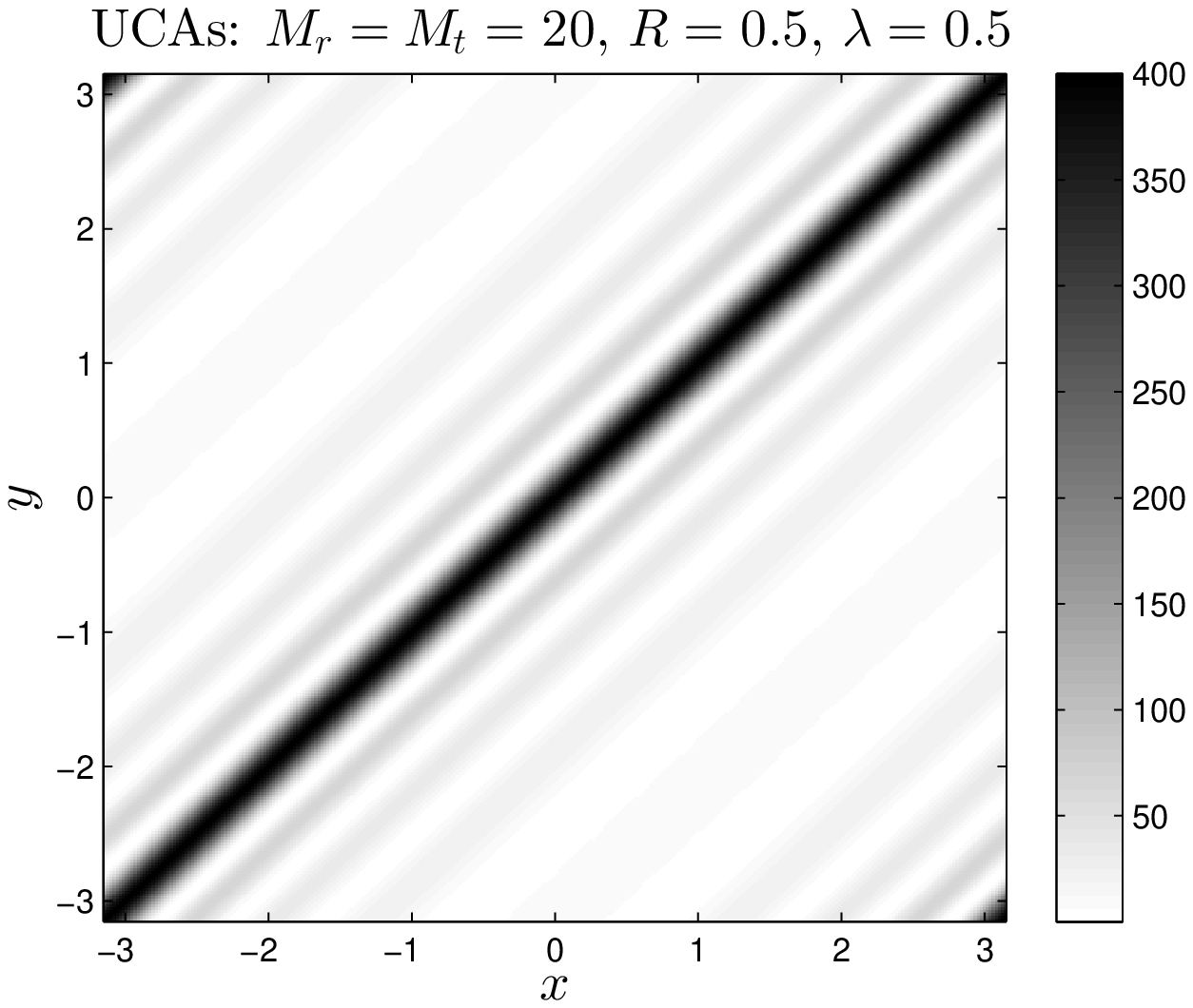}

}\subfloat[]{\includegraphics[scale=0.6]{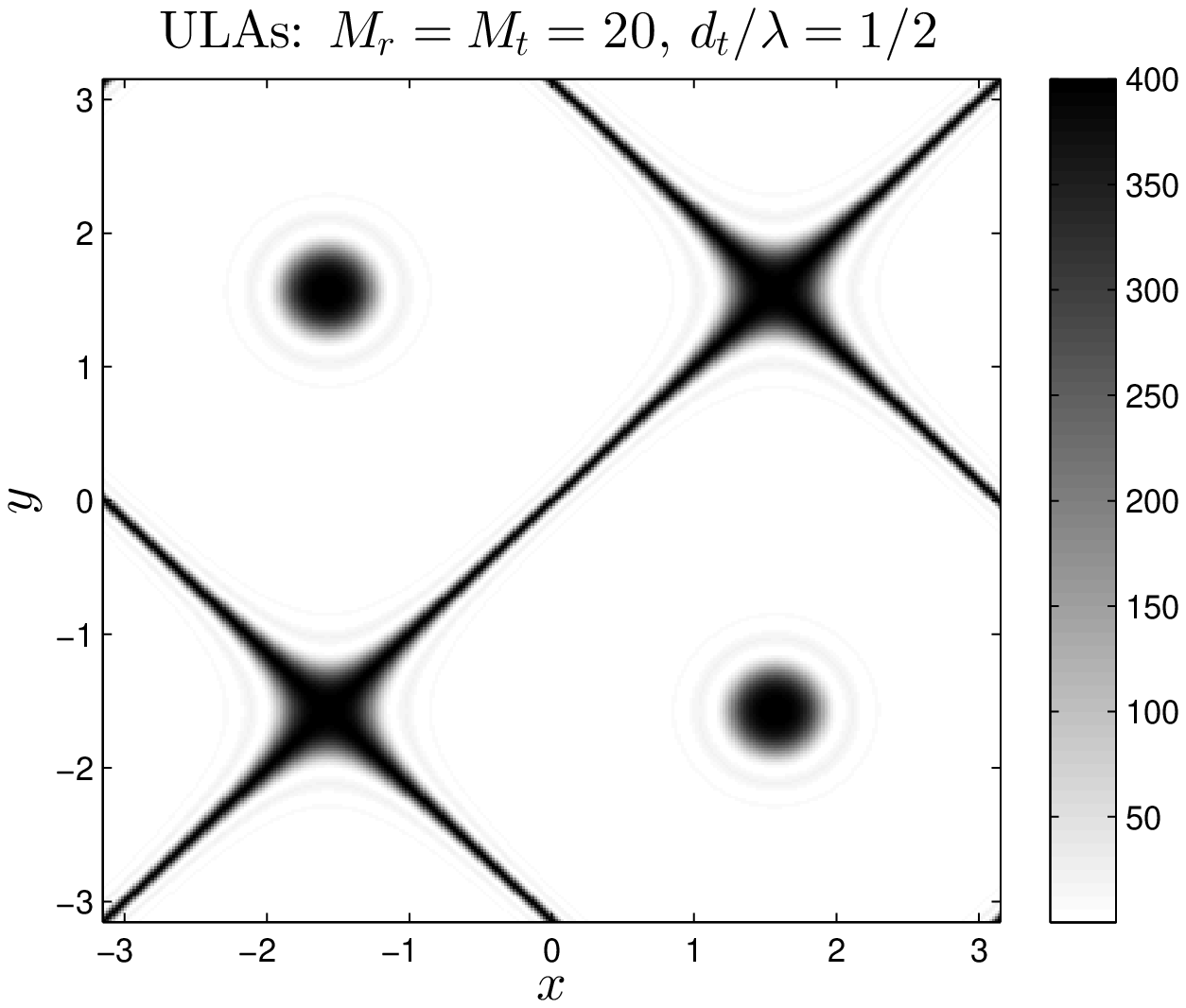}

}

\caption{\label{fig:weird}The function $\left|\varphi_{t\left(r\right)}\left(x,y\left|\mathfrak{T}\left(\mathfrak{R}\right)\right.\right)\right|^{2}$
of Example 1 with respect to $\left(x,y\right)\in\left[-\pi,\pi\right]^{2}$
(just for visualization purposes) for the case of (a) a symmetric
UCA pair and (b) a symmetric ULA pair.}
\end{figure*}
\end{rem}
\begin{example}
Consider a MIMO Radar system equipped (for simplicity) with identical
Uniform Circular Arrays (UCAs) with $M_{r}=M_{t}=20\triangleq M$
antennas, whose positions in the $2$-dimensional plane are defined
as 
\begin{equation}
\begin{bmatrix}x_{l}^{t\left(r\right)}\\
y_{l}^{t\left(r\right)}
\end{bmatrix}\triangleq R\begin{bmatrix}\cos\left(\dfrac{2\pi\left(l-1\right)}{M}\right)\\
\sin\left(\dfrac{2\pi\left(l-1\right)}{M}\right)
\end{bmatrix},
\end{equation}
$l\in\mathbb{N}_{M-1}$, where $R=0.5\, m$. The wavelength utilized
for the communication is chosen as $\lambda=0.5\, m$. Then, Theorem
6 suggests that in order to bound the coherence of the respective
matrix $\boldsymbol{\Delta}$, a uniformly sampled version of which
is available at the receiver's fusion center, we have to specify the
supremum given by \eqref{eq:strange_supremum}, which, of course,
cannot be computed analytically. Fig. \ref{fig:weird}(a) shows the
function $\left|\varphi_{t\left(r\right)}\left(x,y\left|\mathfrak{T}\left(\mathfrak{R}\right)\right.\right)\right|^{2}$
with respect to $\left(x,y\right)\in\left[-\pi,\pi\right]^{2}$. We
observe that this function is periodic in the $2$-dimensional plane,
with period equal to $2\pi$ in each dimension.

Let us now assume that 
\begin{equation}
\mathcal{A}\triangleq\left\{ \left.\left(x,y\right)\in\left[-\dfrac{\pi}{2},\dfrac{\pi}{2}\right]^{2}\right|\eta\leq\left|y-x\right|\leq\pi\right\} ,
\end{equation}
for some $\eta\in\left(0,\pi\right)$, that is, every pair of distinct
angles $\theta_{i}$ and $\theta_{j}$ are such that the magnitudes
of their differences are lying in two non intersecting halfplanes,
one above and one below the hyperplane $y=x$, within a margin of
$\eta$.

\textit{Empirically}, using the graph of $\left|\varphi_{t\left(r\right)}\left(x,y\left|\mathfrak{T}\left(\mathfrak{R}\right)\right.\right)\right|^{2}$,
we can easily specify its supremum over the set $\mathcal{A}$, which
is required in order to bound the coherence of $\boldsymbol{\Delta}$.
Additionally, it is obvious that if $\eta$ is sufficiently large,
the coherence of $\boldsymbol{\Delta}$ will be essentially relatively
small.

Interestingly, this fact draws immediate connections regarding the
coherence between the UCA and the ULA case (see Theorem 3 and Lemma
2). The function $\left|\varphi_{t\left(r\right)}\left(x,y\left|\mathfrak{T}\left(\mathfrak{R}\right)\right.\right)\right|^{2}$
resulting from the application of Theorem 6 for the ULA case is shown
in Fig. \ref{fig:weird}(b). We observe that the sufficient condition
for low coherence of $\boldsymbol{\Delta}$ presented in Lemma 2 can
be immediately validated.

Comparing the two graphs shown in Fig. \ref{fig:weird}, we can infer
that in the UCA case the allowable region for the magnitudes of the
differences of the angle pairs, given by the set $\mathcal{A}$, is
larger than the respective region for the ULA case (Lemma 2). On the
other hand, the values of $\left|\varphi_{t\left(r\right)}\left(x,y\left|\mathfrak{T}\left(\mathfrak{R}\right)\right.\right)\right|^{2}$
for the ULA case in the interior of the allowable region are considerably
smaller than the values of the respective function in the interior
of the respective set for the UCA case. Consequently, regarding our
choice of transmission/reception arrays, there is an obvious trade
- off between ``resolution'' and coherence.\hfill{}\ensuremath{\blacksquare}\end{example}
\begin{rem}
There are other fancier types of $2$-dimensional arrays that can
achieve a much better ``resolution/coherence ratio''. For instance,
we have found experimentally that carefully designed \textit{spiral
arrays}, whose antenna positions constitute a sampling of the usual
\textit{Archimedean Spiral}, possess the advantages of both the aforementioned
cases and, additionally, they exhibit very good asymptotic properties,
guaranteeing low coherence for a sufficiently large number of antennas.
\end{rem}

\section{Coherence of $\boldsymbol{\Delta}$ for the case where\protect \protect \\
 $\theta_{i}\equiv\theta_{j}$ for some $\left(i,j\right)\in\mathbb{N}_{K-1}\times\mathbb{N}_{K-1}$
with $i\neq j$}

Until now, in all the results we have presented, we have assumed that
the target angles $\theta_{k},k\in\mathbb{N}_{K-1}$ are almost surely
distinct. Apparently, this required property of distinctness is far
weaker than (statistical) independence. In fact, our results presented
in detail in the previous sections continue to hold without any modification
for any set of dependent target angles, as long as they are all different
from each other with probability $1$ (recall the Vandermonde/alternant
structure of the matrices $\mathbf{X}_{r}$ and $\mathbf{X}_{t}$
in the proofs of Theorems 3 and 6). Therefore, regarding the coherence
of $\boldsymbol{\Delta}$ for arbitrary $2$-dimensional arrays (and
thus also for ULAs), the only case that needs special treatment is
when $\theta_{i}\equiv\theta_{j}$ for some $\left(i,j\right)\in\mathbb{N}_{K-1}\times\mathbb{N}_{K-1}$
with $i\neq j$, for which we present and prove the following result. 
\begin{thm}
\textbf{\textup{(Coherence for Non - Distinct Angles)}} Consider an
arbitrary $2$-dimensional array transmitter - receiver pair equipped
with $M_{t}$ and $M_{r}$ antennas, respectively. Assume that the
set of angles $\Theta\triangleq\left\{ \theta_{k}\right\} _{k\in\mathbb{N}_{K-1}}$
is such that it can be partitioned as 
\begin{equation}
\Theta=\mathcal{D}\bigcup\mathcal{R}\quad\text{with}\quad\mathcal{D}\bigcap\mathcal{R}\equiv\textrm{Ø},
\end{equation}
where $\mathcal{D}\subseteq\Theta$ contains almost surely distinct
members and 
\begin{equation}
\mathcal{R}\triangleq\left\{ \theta\in\Theta\backslash\mathcal{D}\left|\theta\equiv\rho\text{ a.s.}\text{ for some }\rho\in\mathcal{D}\right.\right\} .
\end{equation}
Let $\boldsymbol{\Delta}$ and $\boldsymbol{\Delta}_{\mathcal{D}}$
be the matrices associated with the angle sets $\Theta$ and $\mathcal{D}$,
respectively. Then, it is true that 
\begin{equation}
\mu\left(U\right)\overset{\text{a.s.}}{\equiv}\mu\left(U_{\mathcal{D}}\right)\quad\text{and}\quad\mu\left(V\right)\overset{\text{a.s.}}{\equiv}\mu\left(V_{\mathcal{D}}\right).
\end{equation}

\end{thm}
\begin{proof}[Proof of Theorem 7]Consider the matrix 
\begin{equation}
\boldsymbol{\Delta}=\mathbf{X}_{r}\mathbf{D}\mathbf{X}_{t}^{\boldsymbol{T}}\in\mathbb{C}^{M_{t}\times M_{r}},
\end{equation}
where $\mathbf{X}_{r}\in\mathbb{C}^{M_{r}\times K}$, $\mathbf{X}_{t}\in\mathbb{C}^{M_{t}\times K}$
and $\mathbf{D}\in\mathbb{C}^{K\times K}$ are defined as in the proof
of either Theorem 3 or Theorem 6. Let the assumptions of the statement
of Theorem 7 hold true. Then, only $\left|\mathcal{D}\right|\triangleq L$
columns of both $\mathbf{X}_{r}$ and $\mathbf{X}_{t}$ will be almost
surely linearly independent, whereas the rest $\left|\mathcal{R}\right|\equiv K-L$
columns of both matrices constitute repetitions of the aforementioned
$L$ ones. As a result, $\boldsymbol{\Delta}$ will be almost surely
of rank $L$.

Let us re-express $\boldsymbol{\Delta}$ as 
\begin{equation}
\boldsymbol{\Delta}=\widetilde{\mathbf{X}}_{r}\mathbf{P}\mathbf{D}\mathbf{P}^{\boldsymbol{T}}\widetilde{\mathbf{X}}_{t}^{\boldsymbol{T}}\equiv\widetilde{\mathbf{X}}_{r}\widetilde{\mathbf{D}}\widetilde{\mathbf{X}}_{t}^{\boldsymbol{T}},\label{eq:D1}
\end{equation}
where $\mathbf{P}\in\mathbb{R}^{K\times K}$ constitutes an appropriately
chosen permutation matrix such that 
\begin{flalign}
\widetilde{\mathbf{X}}_{r} & \triangleq\left[\begin{array}{c|c}
\mathbf{X}_{r}^{\mathcal{D}} & \mathbf{X}_{r}^{\mathcal{R}}\end{array}\right]\quad\text{and}\label{eq:D2}\\
\widetilde{\mathbf{X}}_{t} & \triangleq\left[\begin{array}{c|c}
\mathbf{X}_{t}^{\mathcal{D}} & \mathbf{X}_{t}^{\mathcal{R}}\end{array}\right],\label{eq:D3}
\end{flalign}
where $\mathbf{X}_{r}^{\mathcal{D}}\in\mathbb{C}^{M_{r}\times L}$
(respectively for $\mathbf{X}_{t}^{\mathcal{D}}\in\mathbb{C}^{M_{t}\times L}$)
contains the columns of $\mathbf{X}_{r}$ associated with the almost
surely distinct angles in $\mathcal{D}$ and $\mathbf{X}_{r}^{\mathcal{R}}\in\mathbb{C}^{M_{r}\times K-L}$
(respectively for $\mathbf{X}_{t}^{\mathcal{R}}\in\mathbb{C}^{M_{t}\times K-L}$)
contains the rest of columns of $\mathbf{X}_{r}$, associated with
the rest of the angles in $\mathcal{R}$. Of course, $\widetilde{\mathbf{D}}\in\mathbb{C}^{K\times K}$
constitutes a non - zero - diagonal matrix.

Respectively for $\widetilde{\mathbf{X}}_{t}$, the thin QR decomposition
of $\widetilde{\mathbf{X}}_{r}$ can be expressed as 
\begin{flalign}
\widetilde{\mathbf{X}}_{r} & =\widetilde{\mathbf{V}}_{r}\widetilde{\mathbf{A}}_{r}\nonumber \\
 & \triangleq\left[\begin{array}{c|c}
\mathbf{V}_{r}^{\mathcal{D}} & \mathbf{V}_{r}^{\mathcal{R}}\end{array}\right]\left[\begin{array}{c}
\mathbf{A}_{r}^{\mathcal{D}}\\
\hline \mathbf{0}
\end{array}\right]\nonumber \\
 & \equiv\mathbf{V}_{r}^{\mathcal{D}}\mathbf{A}_{r}^{\mathcal{D}},\label{eq:BETTER}
\end{flalign}
where $\widetilde{\mathbf{V}}_{r}\in\mathbb{C}^{M_{r}\times K}$ is
such that $\widetilde{\mathbf{V}}_{r}^{\boldsymbol{\mathit{H}}}\widetilde{\mathbf{V}}_{r}\equiv\mathbf{I}_{K}$,
$\mathbf{V}_{r}^{\mathcal{D}}\in\mathbb{C}^{M_{r}\times L}$ is such
that $\left(\mathbf{V}_{r}^{\mathcal{D}}\right)^{\boldsymbol{\mathit{H}}}\mathbf{V}_{r}^{\mathcal{D}}\equiv\mathbf{I}_{L}$,
$\widetilde{\mathbf{A}}_{r}\in\mathbb{C}^{K\times K}$ is a specially
structured singular upper triangular matrix (see expressions above)
of rank $L$ and $\mathbf{A}_{r}^{\mathcal{D}}\in\mathbb{C}^{L\times K}$
constitutes a rectangular matrix also of rank $L$.

Let $\left(\mathbf{V}_{r\left(t\right)}^{\mathcal{D}}\right)^{i}\in\mathbb{C}^{1\times L}$,
$i\in\mathbb{N}_{M_{r\left(t\right)}}^{+}$ denote the $i$-th row
of $\mathbf{V}_{r\left(t\right)}^{\mathcal{D}}$. Then, based on \eqref{eq:BETTER}
and using an almost identical procedure as in the first part of the
proof of Theorem 3, we can directly show that the coherence of the
column space of $\boldsymbol{\Delta}$, $U$, is given by 
\begin{flalign}
\mu\left(U\right) & =\frac{M_{r}}{L}\sup_{i\in\mathbb{N}_{N}^{+}}\left\Vert \left(\mathbf{V}_{r}^{\mathcal{D}}\right)^{i}\right\Vert _{2}^{2}.
\end{flalign}
Respectively, the coherence of the row space of of $\boldsymbol{\Delta}$,
$V$, is given by 
\begin{equation}
\mu\left(V\right)=\frac{M_{t}}{L}\sup_{i\in\mathbb{N}_{N}^{+}}\left\Vert \left(\mathbf{V}_{t}^{\mathcal{D}}\right)^{i}\right\Vert _{2}^{2}.
\end{equation}
Now, consider the almost surely rank-$L$ matrix 
\begin{equation}
\boldsymbol{\Delta}_{\mathcal{D}}\triangleq\mathbf{X}_{r}^{\mathcal{D}}\mathbf{D}^{\mathcal{D}}\mathbf{X}_{t}^{\mathcal{D}}\in\mathbb{C}^{M_{t}\times M_{r}},
\end{equation}
where the diagonal matrix $\mathbf{D}^{\mathcal{D}}\in\mathbb{C}^{L\times L}$
is defined in an obvious way (see \eqref{eq:D1} - \eqref{eq:D3}).
Then, by construction of the QR decomposition, it must be true that
\begin{equation}
\mathbf{X}_{r}^{\mathcal{D}}=\mathbf{V}_{r}^{\mathcal{D}}\widehat{\mathbf{A}}_{r}^{\mathcal{D}}\quad\text{and}\quad\mathbf{X}_{t}^{\mathcal{D}}=\mathbf{V}_{t}^{\mathcal{D}}\widehat{\mathbf{A}}_{t}^{\mathcal{D}},
\end{equation}
where, $\widehat{\mathbf{A}}_{t}^{\mathcal{D}}\in\mathbb{C}^{L\times L}$
and $\widehat{\mathbf{A}}_{r}^{\mathcal{D}}\in\mathbb{C}^{L\times L}$
constitute almost surely full rank matrices. Recall \eqref{eq:RECALL}
and adapt it appropriately for expressing the coherences of the column
and row spaces of $\boldsymbol{\Delta}_{\mathcal{D}}$, $U_{\mathcal{D}}$
and $V_{\mathcal{D}}$, respectively. Our claims follow.\end{proof} 
\begin{rem}
Indeed, in many applications (Radar and non - Radar), it is very likely
that the situations Theorem 7 is dealing with may correspond to events
of zero measure. However, we feel that stating this result is important,
in order for the coherence analysis we presented in the previous chapters
to be entirely complete. Also, interestingly, the simple implication
of Theorem 7 is somewhat not intuitive: As the difference between
two angles is decreasing, the coherence of $\boldsymbol{\Delta}$
becomes worse (that is, increases), but if the difference is exactly
zero, then the coherence of $\boldsymbol{\Delta}$ is identical to
the coherence of the variation of $\boldsymbol{\Delta}$ resulting
from simply removing the information corresponding to any of the two
identical angles. Clearly, this result does not in any way contradict
the ones presented in the previous sections. 
\end{rem}

\section{Discussion and Some Simulations}

In this section, we present some simulations, validating our main
results presented above and discuss some further implications of our
respective theoretical analysis.

For simplicity, we consider a MIMO radar system equipped with identical
ULAs for transmission and reception, with $d_{r}\equiv d_{t}=\lambda/2$
and $M_{r}\equiv M_{t}\triangleq M$. We also consider $K=4$ targets
in the far field, with angles independently distributed in $\left(-\dfrac{\pi}{2},\dfrac{\pi}{2}\right)$.
\begin{figure}
\centering\includegraphics[scale=0.6]{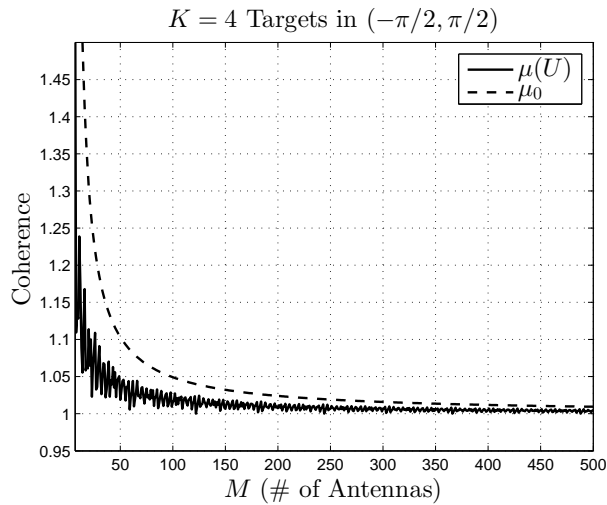}\caption{\label{fig:simple_bound}Behavior of the coherence of $\boldsymbol{\Delta}$
and the proposed bound with respect to the number of antennas $M$.}
\end{figure}

In this synthetic example, in order to demonstrate the validity of
Theorem 3, we assume that the target angles are known apriori. By
Lemma 2, we know that the condition $\xi\neq0$ will always hold and
consequently the asymptotics of Theorem 3 must also hold true. Of
course, $\xi$ can be computed either using \eqref{eq:xixixixi} or
\eqref{eq:xi_worstcase} (for some sufficiently chosen value of $\eta$),
with the latter producing a worst - case bound regarding the coherence
of the associated matrix $\boldsymbol{\Delta}$. Also, in this example,
we obviously have $\mu\left(U\right)\equiv\mu\left(V\right)$.

Using \eqref{eq:xixixixi} for the computation of $\xi$, Fig. \ref{fig:simple_bound}
shows the behavior of both the coherence of $\boldsymbol{\Delta}$
and its bound $\mu_{0}$, which results by directly applying Theorem
3, as a function of the number of transmission/reception antennas
$M$. We observe that, clearly, the proposed bound is tight and it
tracks the convergence of the coherence of $\boldsymbol{\Delta}$
to unity very accurately. We should mention here that for very small
angle differences, the bound becomes somewhat looser. However, our
numerical simulations have shown that our bound constitutes a very
reasonable coherence estimate in all cases. 
\begin{figure}
\centering\includegraphics[scale=0.6]{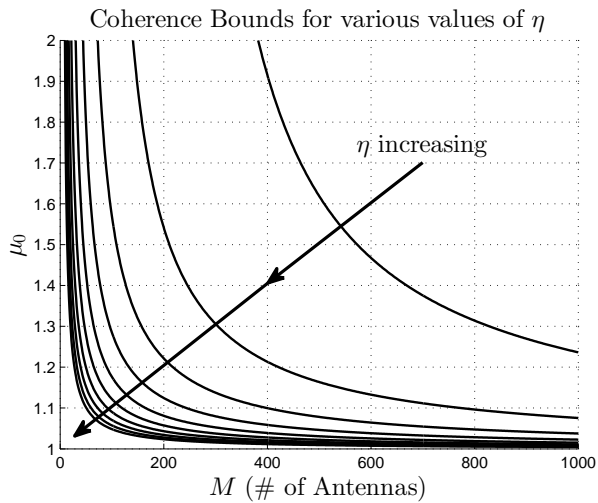}\caption{\label{fig:ETA_bound}Our proposed bound with respect to the number
of antennas $M$, for different values of $\eta$, according to Lemma
2.}
\end{figure}

Apparently, in any realistic situation, the actual values of the target
angles are unknown. Assuming generically that each of the angle differences
belongs to a set given by \eqref{eq:SET_LEMMA2} for some $\eta\in\left(0,\dfrac{\pi}{2}\right)$
-$\eta$ depends on the specific radar application-, we can invoke
Lemma 2 in order to bound the coherence of $\boldsymbol{\Delta}$
in this more general case. Fig. \ref{fig:ETA_bound} depicts a number
of bounds produced by Lemma 2 for various values of $\eta$, as functions
of the number of antennas $M$. One can observe that, as the value
of $\eta$ increases, the respective coherence bound converges much
faster to unity, therefore increasing our confidence that the performance
of matrix completion will be satisfactory for a relatively smaller
number of transmission/reception antennas.
\begin{figure*}
\centering\subfloat[]{\includegraphics[scale=0.6]{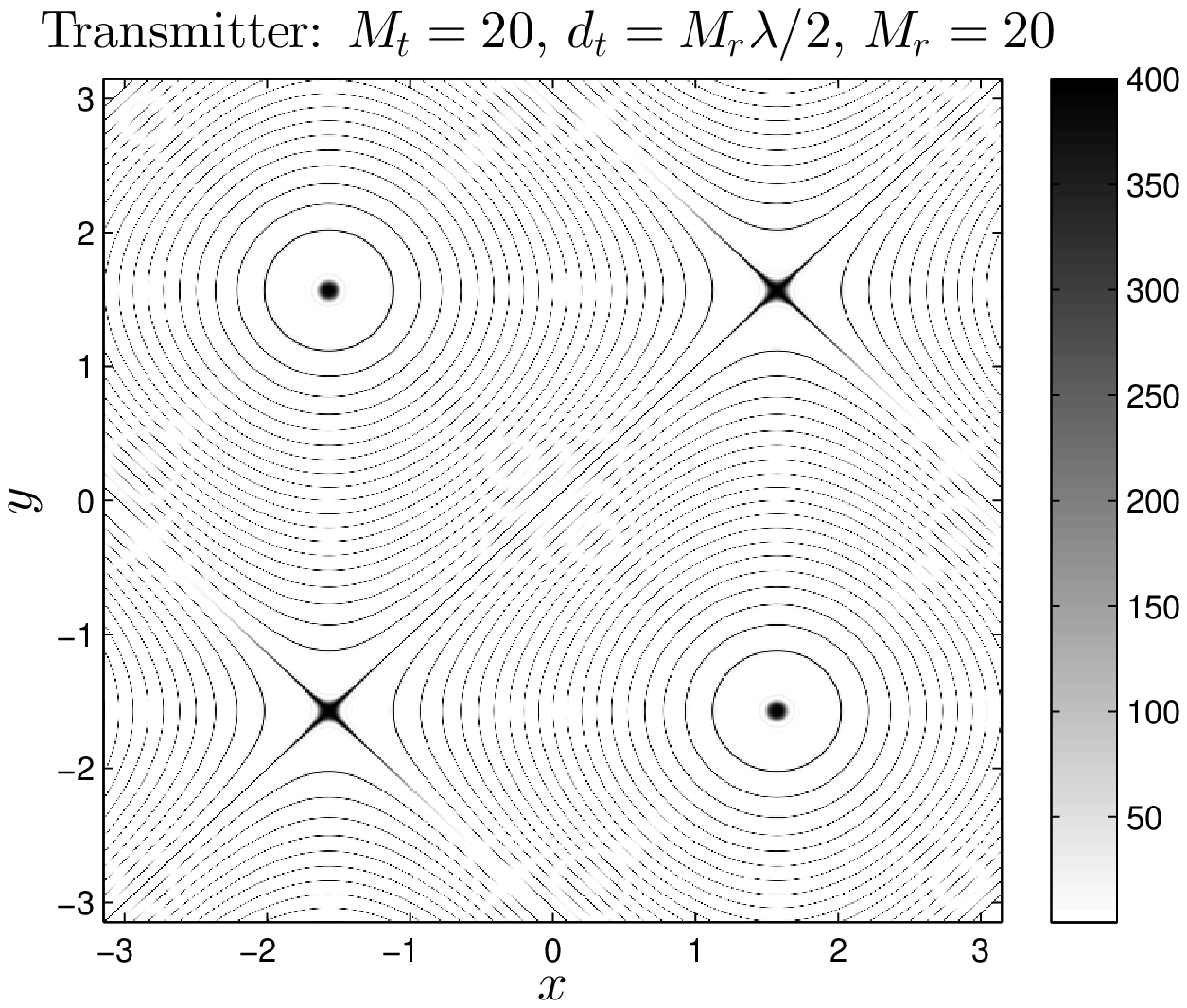}

}\subfloat[]{\includegraphics[scale=0.6]{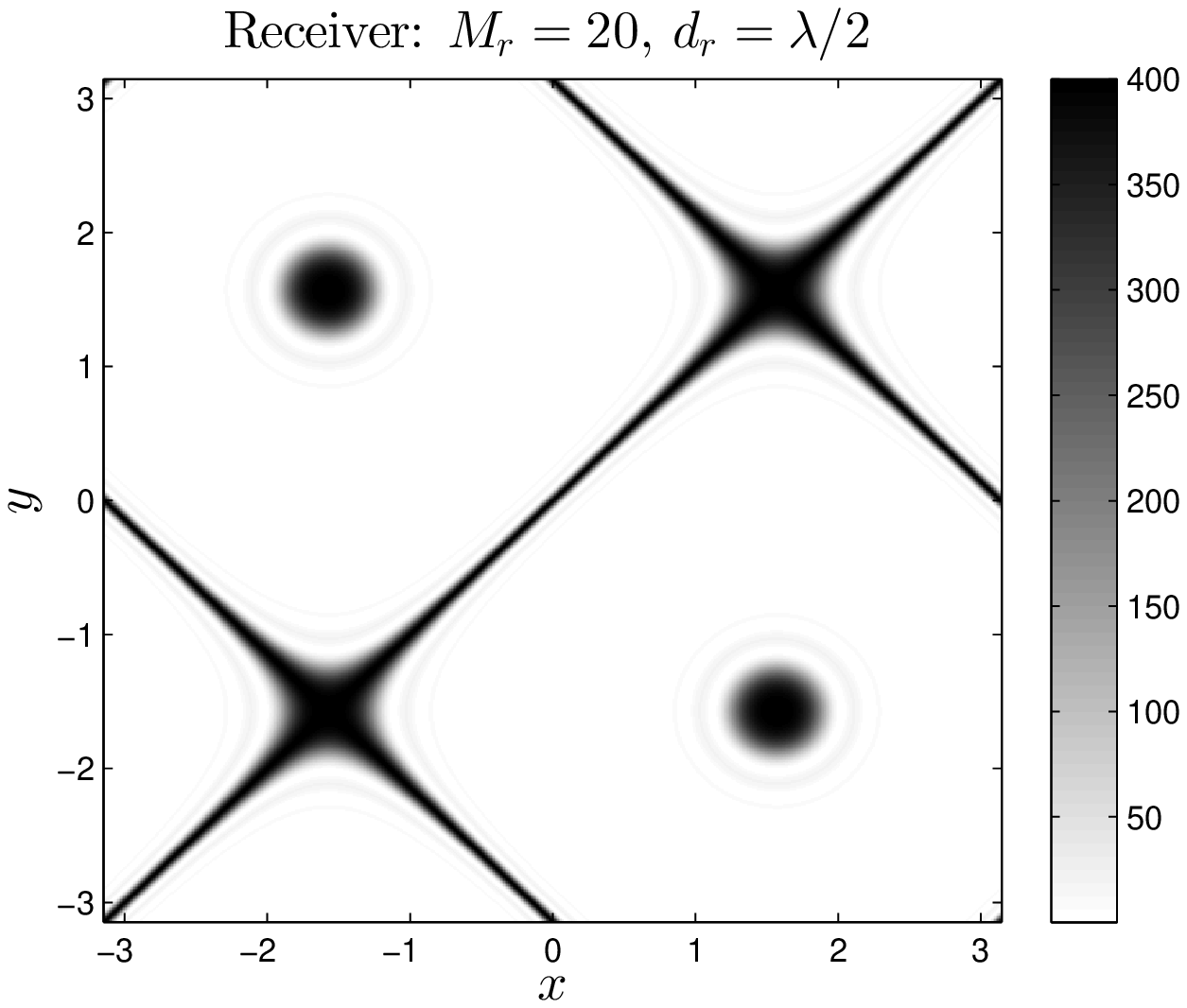}

}

\caption{\label{fig:weird-1}The function $\left|\varphi_{t\left(r\right)}\left(x,y\left|\mathfrak{T}\left(\mathfrak{R}\right)\right.\right)\right|^{2}$
with respect to $\left(x,y\right)\in\left[-\pi,\pi\right]^{2}$ for
(a) the transmission array ($\mathfrak{T}$) and (b) the reception
array ($\mathfrak{R}$) of the high degree of freedom MIMO radar considered
in this case.}
\end{figure*}

However, we should stress that Lemma 2 essentially produces worst
case bounds, which might be too pessimistic for specific applications.
If this is the case, one could use Theorem 6, which provides greater
versatility for deriving more accurate bounds for more specific target
configurations and for given number of antennas or, alternatively,
one could derive variants of Lemma 2 for special target configurations.

Another important ULA configuration in MIMO radar systems results
by choosing $d_{t}=M_{r}\lambda/2$ and $d_{r}=\lambda/2$. It is
well known that this particular choice for the transmitter antenna
spacing increases the degrees of freedom of the system to $M_{t}M_{r}$
using only $M_{t}+M_{r}$ physical antenna elements and consequently,
higher spatial resolution can be obtained (see, for instance, \cite{ChenVaidyanathan2008}).
If we are interested in the performance of matrix completion, as far
as $\mu\left(V\right)$ is concerned, associated solely with the reception
array, we can directly use Theorem 3 in conjunction with Lemma 2 as
explained above in order to derive the respective coherence bound.
However, the behavior of $\mu\left(U\right)$ is much more complicated
and the application of Theorem 3 leads to an unreasonably involved
analysis. But one could invoke Theorem 6 in this case. Fig. \ref{fig:weird-1}(a)
shows the function $\left|\varphi_{t}\left(x,y\left|\mathfrak{T}\right.\right)\right|^{2}$
with respect to $\left(x,y\right)\in\left[-\pi,\pi\right]^{2}$. Assuming
that the angles are independently and uniformly distributed in $\left(-\dfrac{\pi}{2},\dfrac{\pi}{2}\right)$
and that the contours of $\left|\varphi_{t}\left(x,y\left|\mathfrak{T}\right.\right)\right|^{2}$
are sufficiently dense (for $M_{r}$ properly chosen), then the respective
values of the aforementioned function will be essentially small, since,
technically speaking, the probability of any $\left(\theta_{i},\theta_{j}\right)$
belonging to the union of the strict subsets of $\left(-\dfrac{\pi}{2},\dfrac{\pi}{2}\right)^{2}$
corresponding to the contours of $\left|\varphi_{t}\left(x,y\left|\mathfrak{T}\right.\right)\right|^{2}$
shown in Fig. \ref{fig:weird-1}(a) will be essentially zero. Therefore,
a low value for $\mu\left(U\right)$ is guaranteed with very high
probability.

\section{Conclusion}

In this paper, we have presented a detailed analysis regarding the
recoverability of the data matrix in colocated MIMO radar systems
via convex optimization (matrix completion), for a number of commonly
used array configurations. We showed that, in most cases, the data
matrix is indeed recoverable from a minimal number of observations,
as long as the number of transmission and reception antennas is sufficiently
large and under common assumptions on the DOAs of the targets. Consequently,
the matrix completion approach for reducing the sampling requirements
in colocated MIMO radar is indeed theoretically robust and also appealing
for practical consideration in real world applications. However, although
we have explicitly shown that the choice of ULAs for transmission
and reception leads to asymptotically and approximately optimal coherence
of the data matrix, the important and more general problem of optimally
choosing the array topologies for minimizing matrix coherence still
remains open, clearly suggesting new directions for further research.

\section*{Appendix A\protect \\
Proof of Lemma 1}

Let the hypothesis of the statement of Lemma 1 hold true. In the following,
we consider only the column space of $\mathbf{M}$, $U$, since, for
$V$ the procedure is identical. First, consider the set $\mathcal{G}\triangleq\left\{ \left(i,j\right)\in\mathbb{N}_{N}^{+}\times\mathbb{N}_{N}^{+}\left|i\equiv j\right.\right\} $.
Then, 
\begin{flalign}
 & \sup_{\left(i,j\right)\in\mathcal{G}}\left|\frac{N}{r}\left\langle \mathbf{e}_{i},\mathbf{P}_{U}\mathbf{e}_{j}\right\rangle -\mathds{1}_{i=j}\right|\nonumber \\
\equiv & \sup_{i\in\mathbb{N}_{N}^{+}}\frac{N}{r}\left\Vert \mathbf{P}_{U}\mathbf{e}_{i}\right\Vert _{2}^{2}-1\le\mu_{0}-1.
\end{flalign}
Next, consider the set $\mathcal{H}\triangleq\mathbb{N}_{N}^{+}\times\mathbb{N}_{N}^{+}-\mathcal{G}$.
In this case, 
\begin{flalign}
 & \sup_{\left(i,j\right)\in\mathcal{H}}\left|\frac{N}{r}\left\langle \mathbf{e}_{i},\mathbf{P}_{U}\mathbf{e}_{j}\right\rangle -\mathds{1}_{i=j}\right|\nonumber \\
\equiv & \sup_{\left(i,j\right)\in\mathcal{H}}\frac{N}{r}\left|\left\langle \mathbf{P}_{U}\mathbf{e}_{i},\mathbf{P}_{U}\mathbf{e}_{j}\right\rangle \right|\nonumber \\
\le & \sup_{\left(i,j\right)\in\mathcal{H}}\frac{N}{r}\left\Vert \mathbf{P}_{U}\mathbf{e}_{i}\right\Vert _{2}\left\Vert \mathbf{P}_{U}\mathbf{e}_{j}\right\Vert _{2}\le\mu_{0}.
\end{flalign}
In all cases, it is true that 
\begin{equation}
\mu_{s}\left(U\right)\le\mu_{0}\quad\text{and}\quad\mu_{s}\left(V\right)\le\mu_{0}.
\end{equation}
Consequently, the assumption $\mathbf{A2}$ is trivially satisfied
with 
\begin{equation}
\max\left\{ \mu_{s}\left(U\right),\mu_{s}\left(V\right)\right\} \le\mu_{0}
\end{equation}
and then we can take $\mu_{0}^{s}\le\mu_{0}\sqrt{r}$. Also, $\mu_{1}\le\mu_{0}\sqrt{r}$
\cite{Candes&Recht2009,Candes&Tao2010}. Choosing $\mu\triangleq\max\left\{ \mu_{0}^{s},\mu_{1}\right\} $
completes the proof.\hfill{}\ensuremath{\qed}

\section*{Appendix B\protect \protect \\
 Proof of Lemma 2}

Let the hypotheses of the statement of Lemma 2 hold true and consider
a bivariate function $f:\mathcal{A}\rightarrow\mathcal{F}\subset\left(0,2\right)$
defined as 
\begin{equation}
f\left(x,y\right)\triangleq\left|\sin\left(x\right)-\sin\left(y\right)\right|,
\end{equation}
obviously bounded from above and below as 
\begin{equation}
\inf_{\left(x,y\right)\in\mathcal{A}}f\left(x,y\right)\leq f\left(x,y\right)\leq\sup_{\left(x,y\right)\in\mathcal{A}}f\left(x,y\right).\label{eq:ineq1}
\end{equation}
In the following, our goal will be to explicitly specify the infimum
and the supremum appearing in \eqref{eq:ineq1}, or equivalently the
set $\mathcal{F}$.

Consider the case where $\left(x,y\right)\in\left\{ \left(x,y\right)\in\mathcal{A}\left|y>x\right.\right\} \triangleq\mathcal{A}_{+}$.
Thus, 
\begin{equation}
f\left(x,y\right)=\sin\left(y\right)-\sin\left(x\right)\triangleq f_{+}\left(x,y\right).
\end{equation}
Also, let us define the set 
\begin{equation}
\mathcal{L}_{\beta}\triangleq\left\{ \left(x,y\right)\in\mathcal{A}_{+}\left|y=x+\beta,\beta\in\left[\eta,\pi-\eta\right]\right.\right\} .
\end{equation}
Then, restricted to $\mathcal{L}_{\beta}$, the function $f_{+}\left(x,y\right)$
can be expressed as 
\begin{equation}
f_{+}\left(x,x+\beta\right)=\sin\left(x+\beta\right)-\sin\left(x\right)\triangleq f_{+}\left(x;\beta\right),
\end{equation}
where $x\in\left[-\dfrac{\pi}{2},\dfrac{\pi}{2}-\beta\right]$. By
a simple second derivative test, it can be shown that $f_{+}\left(x;\beta\right)$
is a strictly concave function in the set that it is defined. It is
then very easy to show that $f_{+}\left(x;\beta\right)$ is maximized
at $\widehat{x}\equiv-\dfrac{\beta}{2}$ over the set $\left[-\dfrac{\pi}{2},\dfrac{\pi}{2}-\beta\right]$,
whereas its infimum occurs at either of the boundaries of the feasible
set, say, at $\widetilde{x}\equiv-\dfrac{\pi}{2}$. Thus, trivially,
the maximum of $f_{+}\left(x,y\right)$ over the set $\mathcal{L}_{\beta}$
occurs at the point 
\begin{flalign}
\left(\widehat{x},\widehat{y}\right) & \equiv\left(-\dfrac{\beta}{2},\dfrac{\beta}{2}\right),\quad\text{with}\\
\max_{\left(x,y\right)\in\mathcal{L}_{\beta}}f_{+}\left(x,y\right) & \equiv2\sin\left(\dfrac{\beta}{2}\right),\quad\forall\,\beta\in\left[\eta,\pi-\eta\right],
\end{flalign}
whereas its infimum occurs at 
\begin{flalign}
\left(\widetilde{x},\widetilde{y}\right) & \equiv\left(-\dfrac{\pi}{2},\beta-\dfrac{\pi}{2}\right),\quad\text{with}\\
\inf_{\left(x,y\right)\in\mathcal{L}_{\beta}}f_{+}\left(x,y\right) & \equiv1-\cos\left(\beta\right),\quad\forall\,\beta\in\left[\eta,\pi-\eta\right].
\end{flalign}
Apparently, both the maximum and infimum of $f_{+}\left(x,y\right)$
over $\mathcal{L}_{\beta}$ constitute strictly increasing functions
of $\beta$, which directly implies that 
\begin{equation}
\sup_{\left(x,y\right)\in\mathcal{A}_{+}}f_{+}\left(x,y\right)\equiv2\sin\left(\dfrac{\pi-\eta}{2}\right)\equiv2\cos\left(\dfrac{\eta}{2}\right)
\end{equation}
and 
\begin{equation}
\inf_{\left(x,y\right)\in\mathcal{A}_{+}}f_{+}\left(x,y\right)\equiv1-\cos\left(\eta\right).
\end{equation}
Further, observe that $f\left(x,y\right)$ for $\left(x,y\right)\in\mathcal{A}$
is symmetric with respect to its main diagonal. Consequently, it must
be true that 
\begin{flalign}
\sup_{\left(x,y\right)\in\mathcal{A}}f\left(x,y\right) & \equiv2\cos\left(\dfrac{\eta}{2}\right)\quad\text{and}\\
\inf_{\left(x,y\right)\in\mathcal{A}}f\left(x,y\right) & \equiv1-\cos\left(\eta\right),\quad\forall\,\eta\in\left(0,\dfrac{\pi}{2}\right].
\end{flalign}
Therefore, $\forall\,\left(x,y\right)\in\mathcal{A}$, 
\begin{equation}
f\left(x,y\right)\in\left[1-\cos\left(\eta\right),2\cos\left(\dfrac{\eta}{2}\right)\right]\equiv\mathcal{F}.\label{eq:okokok}
\end{equation}
Now, substituting $x$ and $y$ with $\theta_{i}$ and $\theta_{j}$,
respectively, \eqref{eq:okokok} yields 
\begin{equation}
\dfrac{1-\cos\left(\eta\right)}{2}\leq\dfrac{1}{2}\left|\sin\left(\theta_{i}\right)-\sin\left(\theta_{j}\right)\right|\leq\cos\left(\dfrac{\eta}{2}\right),
\end{equation}
holding true $\forall\,\left(i,j\right)\in\mathbb{N}_{K-1}\times\mathbb{N}_{K-1}$
with $i\neq j$ and $\forall\,\eta\in\left(0,\dfrac{\pi}{2}\right)$,
where, in general, 
\begin{flalign}
\dfrac{1-\cos\left(\eta\right)}{2} & \in\left(0,\dfrac{1}{2}\right]\quad\text{and}\\
\cos\left(\dfrac{\eta}{2}\right) & \in\left[\dfrac{\sqrt{2}}{2},1\right).
\end{flalign}
As a result, for a fixed $\eta$ (recall the definition of $\xi$
in the statement of Theorem 3), 
\begin{flalign}
\xi & =\min_{\substack{\left(i,j\right)\in\mathbb{N}_{K-1}\times\mathbb{N}_{K-1}\\
i\neq j
}
}g\left(\dfrac{1}{2}\left|\sin\left(\theta_{i}\right)-\sin\left(\theta_{j}\right)\right|\right)\nonumber \\
 & =\min\left\{ \dfrac{1-\cos\left(\eta\right)}{2},g\left(\cos\left(\dfrac{\eta}{2}\right)\right)\right\} \nonumber \\
 & \equiv\min\left\{ \dfrac{1-\cos\left(\eta\right)}{2},1-\cos\left(\dfrac{\eta}{2}\right)\right\} \nonumber \\
 & \equiv1-\cos\left(\dfrac{\eta}{2}\right)\in\left(0,\dfrac{2-\sqrt{2}}{2}\right],
\end{flalign}
thus completing the proof.\hfill{}\ensuremath{\qed}

 \bibliographystyle{IEEEtran}
\bibliography{IEEEabrv}

\end{document}